\documentclass[journal,twoside,web]{ieeecolor}
\usepackage{generic}
\usepackage{cite}
\usepackage{amsmath,amssymb,amsfonts}
\usepackage{algorithmic}
\usepackage{graphicx}
\usepackage{algorithm,algorithmic}
\usepackage{bm}
\usepackage[caption=false,font=footnotesize]{subfig}
\usepackage{hyperref}
\hypersetup{hidelinks}
\usepackage{textcomp}
\usepackage{array}
\usepackage{booktabs}

\usepackage{amsthm}
\newtheorem{thm}{Theorem}
\newtheorem{prop}{Proposition}
\newtheorem{lem}{Lemma}
\newtheorem{corollary}{Corollary}
\theoremstyle{remark}
\newtheorem{defn}{Definition}
\newtheorem{prob}{Problem}
\newtheorem{remark}{Remark}

\newcolumntype{C}[1]{>{\centering\arraybackslash}m{#1}}
\newcommand{\R}{\mathbb{R}}
\newcommand{\G}{\mathcal{G}}
\newcommand{\Oc}{\mathcal{O}}

\newcommand{\M}{\mathcal{M}}
\usepackage{multirow}
\usepackage[caption=false,font=footnotesize]{subfig}
\usepackage{bbm}

\newcommand{\X}{\mathcal{X}}

\newcommand{\ub}{\bm{u}}

\newcommand{\yb}{\bm{y}}

\newcommand{\xb}{\bm{x}}
\newcommand{\Gb}{\bm{G}}
\newcommand{\Ob}{\bm{O}}

\DeclareMathOperator*{\argmin}{arg\,min}

\begin{document}

\title{Robustly Invertible Nonlinear Dynamics and the BiLipREN: From Inversion-Based Control to Generative Trajectory Modelling}

\author{Yurui Zhang, Ruigang Wang and Ian R. Manchester
\thanks{This work was supported in part by the Australian Research Council through projects DP230101014 and IH210100030.}
\thanks{The authors are with the Australian Centre for Robotics (ACFR), and the School of Aerospace, Mechanical and Mechatronic Engineering, The University of Sydney, Australia, {\tt\small \{yurui.zhang, ruigang.wang, ian.manchester\}@sydney.edu.au}}.}

\maketitle

\begin{abstract}
This paper proposes a new notion of robust invertibility for nonlinear dynamical systems, and introduces constructive parameterizations of recurrent neural network which are robustly invertible by design. We define robust invertibility as the existence of a causal inverse system such that both the forward and inverse systems are contracting and have bounded incremental input-output gains (the system is bi-Lipschitz), implying that both forward prediction and input reconstruction are robust to signal perturbations and initial-state mismatch. We  construct robustly invertible recurrent models via series composition of static orthogonal layers and dynamic layers satisfying a strong input-output monotonicity property, and provide a differentiable neural network parameterizations in the form of the bi-Lipschitz recurrent equilibrium network (BiLipREN). Additionally, composition with  dynamic orthogonal layers yields a nonlinear minimum-phase/all-pass  (a.k.a. inner--outer) factorization. We illustrate  the utility of the framework through a series of application examples in data-driven internal model control, dynamic surrogate loss learning, and signal-space normalizing flows, illustrating its utility for robust control, trajectory optimization, and generative modeling of complex trajectory distributions.
\end{abstract}

\begin{IEEEkeywords}
Stability of Nonlinear Systems, Contraction Theory, Neural Networks, Inverse Optimal Control, Generative Modelling.
\end{IEEEkeywords}

\section{Introduction}
\label{sec:introduction}
\IEEEPARstart{T}{he} notion of dynamic system invertibility plays a fundamental role in control theory and applications. In the absence of any uncertainty, an ideal inverse of a system provides a perfect feedforward controller, while in the feedback setting robust  control was famously characterized by Zames in terms of existence of an approximate inverse \cite{zames1981feedback}. 

Invertible nonlinear functions play a similarly fundamental role in machine learning, especially generative modelling through the concept of a \textit{normalizing flow} \cite{papamakarios2021normalizing}: an invertible mapping between a simple probability distribution (a multivariate normal) and a complex distribution (e.g. the distribution of images of a certain object).

In this paper, we introduce new notions of invertibility for nonlinear dynamical systems and new recurrent neural network model classes with guaranteed robust invertibility, and illustrate their utility in examples spanning classical control techniques to modern generative modelling.

\subsection{System Inversion and Control}

Many practical techniques in control design and signal processing are based in some way on dynamic system inversion. System inverses are used for feedforward compensation in precision motion control (e.g. \cite{markusson2001iterative, van2018inversion, bolderman2024physics}) and for digital predistortion compensation in electronics (e.g. \cite{ghannouchi2009behavioral, liu2006augmented, tanovic2018equivalent}). System inversion is used a component in feedback schemes such as nonlinear dynamic inversion for flight control (e.g. \cite{reiner1996flight, miller2011nonlinear, sieberling2010robust, wang2017robust}) and internal model control schemes which have been widely applied in the process industries (e.g. \cite{garcia1982internal, economou1986internal}). 
Other works such as \cite{celani2010output, wang2017robust} developed robust output feedback regulators for certain classes of invertible nonlinear MIMO systems.

A linear system $\yb = \bm{G}(z)\ub$ with all poles and zeros strictly inside the unit circle (or left-half-plane in continuous-time) is called a minimum-phase system. For such systems, both the $\bm G(z)$ and its inverse $\bm G(z)^{-1}$ are stable and have finite gain. Many real-world systems are not minimum-phase, but any causal stable system $\bm P(z)$ can be factored into a product of a all-pass (\textit{a.k.a} inner) system $\bm \Ob(z)$ -- containing time-delays and unstable zeros -- and a minimum-phase (\textit{a.k.a.} outer) system $\bm G(z)$ \cite{OppenheimSchafer2010}:
\[
\bm P(z) = \bm O (z)\bm G(z).
\]
Such factorizations are classical and have many applications in signal processing \cite{OppenheimSchafer2010}, e.g. for spectral equalization of channels, while classical feedback control schemes such as the Smith predictors \cite{smith1957closer, laughlin1987smith} and internal model control \cite{garcia1982internal} make use of this factorization, and have also been extended to nonlinear systems, see e.g. \cite{van2018l2}.

\begin{figure*}[!t]
\centering
\renewcommand{\arraystretch}{1.2}
\begin{tabular}{C{1.4cm}|C{0.42\textwidth}|C{0.44\textwidth}}
\hline
\textbf{Mapping}
&
\textbf{Finite-dimensional Vector Space} $\mathbb{R}^n$
&
\textbf{Infinite-dimensional Signal Space} $\ell^n$
\\
\hline
\textbf{Linear}
&
{\large Matrix transformation}

\vspace{0.5em}

\includegraphics[height=3cm]{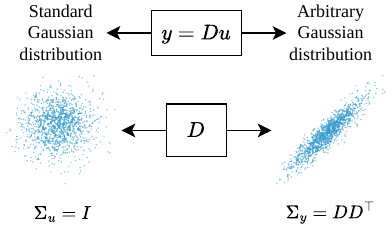}

 $D$ is an invertible matrix
&
{\large Spectral factorization}

\vspace{0.5em}

\includegraphics[height=3cm]{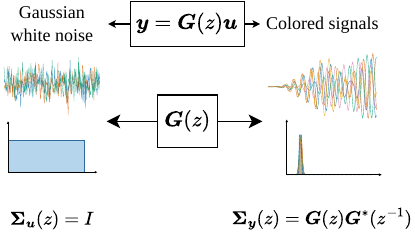}

$\bm G(z)$ is a bi-proper minimum-phase LTI system
\\
\hline
\textbf{Nonlinear}
&
{\large Normalizing flows}

\vspace{0.5em}

\includegraphics[height=3cm]{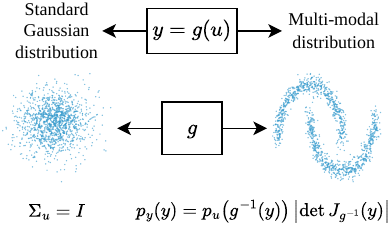}

$g$ is an invertible function
&
{\large The proposed BiLipREN}

\vspace{0.5em}

\includegraphics[height=3cm]{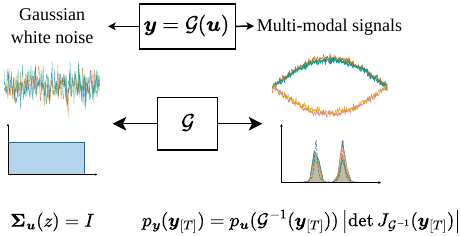}

$\G$ is a robustly invertible system
\\
\hline
\end{tabular}
\caption{A unified view of spectral factorization and generative modelling via linear and nonlinear mappings of Gaussian distributions in vector and signal space.}\label{fig:comparison}
\end{figure*}

\subsection{Surrogate Loss Learning}
In this paper we also investigate the utility of invertible dynamic models for \textit{surrogate loss learning}, i.e. learning an approximation to an intractable or expensive black-box loss (or reward) function for the purposes of optimization. 

The classical \textit{response surface methodology} is widely used in process optimization \cite{myers2016response}. Here the relationship between (vector) inputs $u$ and a loss or reward function $J(u)$ is typically modelled as a low-order polynomial, although other representations such as radial basis functions have been proposed \cite{mcdonald2007global}. In the Bayesian optimization framework (e.g. \cite{jones1998efficient}) a probabilistic surrogate such as a Gaussian process it commonly used, and its predictive distribution is combined with an acquisition function to determine subsequent test points \cite{snoek2012practical}. Surrogate loss learning has also been applied in deep learning for optimizing non-differentiable losses (e.g. \cite{grabocka2019learning, patel2020learning}) and in reward modelling for language models \cite{leike2018scalable}. A class of ``easily optimizable'' neural surrogate losses is the input convex neural network (ICNN) \cite{amos2017input},

In \cite{wangmonotone}, Polyak--\L{}ojasiewicz networks (PLNets) were introduced in the form by representing the surrogate loss as $J=|g(u)|^2$ where $g$ bi-Lipschitz function and $u\in\mathbb R^n$. This construction implies that $J$ satisfies the PL condition guaranteeing global linear convergence of gradient descent \cite{polyak1963gradient, lojasiewicz1963topological, karimi2016linear}, and furthermore allows rapid computation of the surrogate's global minimum via a splitting algorithm. In a sense, they fit surrogate losses which ``easy to optimize'' but less restrictive than ICNNs. In this paper we show how this framework can be extended from loss functions of a vector of fixed dimension to loss functions of an arbitrary-length trajectory, motivated by problems in optimal control.

\subsection{Invertible Mappings, Spectral Factorization, and Generative Modelling}
Invertible mappings also play a central role in generative modelling, i.e. learning to sample vectors or signals from a complex distribution. A common approach is to sample from a simple distribution (e.g. a Gaussian) and then pass each sample through a learnt invertible transformation. As illustrated in Figure~\ref{fig:comparison}, this approach can be classified according to two dimensions: linear versus nonlinear mappings, and finite-dimensional vectors versus infinite-dimensional signals. 
In the simplest linear vector setting, the transformation is an invertible matrix \(D\) obtained by factorizing the covariance:
\begin{equation}
    \Sigma_y = DD^\top,
    \qquad
    y=Du,
    \qquad
    u\sim\mathcal N(0,I).
\end{equation}
The invertible matrix \(D\) shapes a standard Gaussian vector into a Gaussian vector with covariance \(\Sigma_y\).

The corresponding construction in the linear signal setting can be related to the classical problem of \emph{spectral factorization}, which extends covariance factorization from random vectors to stationary processes. 
For a discrete-time linear SISO system, minimum phase means that all poles and zeros lie strictly inside the unit circle \cite{OppenheimSchafer2010}. 
If the system bi-proper (relative degree zero), then its inverse is stable and causal. Let \(\yb\) be a zero-mean stationary process with rational and strictly positive-definite power spectral density (PSD) \(\bm{\Sigma}_{\yb}(z)\). 
Under standard assumptions, there exists a bi-proper minimum-phase system \(\bm{G}(z)\) such that
\begin{equation}\label{eq:spec-fact}
    \bm{\Sigma}_{\yb}(z)
    =
    \bm{G}(z)\bm{G}^*(z^{-1}).
\end{equation}
This factorization corresponds to the generative model
\begin{equation}
    \yb=\bm{G}(z)\ub,
\end{equation}
where the standard Gaussian white-noise process \(u_t\sim\mathcal N(0,I)\) is filtered through \(\bm{G}(z)\)  so that the output process \(\yb\) has the desired spectrum.

In the nonlinear vector setting, \emph{normalizing flows} replace the linear transformation \(D\) with invertible nonlinear maps of various forms. 
They transform samples from a standard Gaussian, into samples from a more complex target distribution which may be multi-modal
\cite{papamakarios2021normalizing,lipman2023flow}. 
Let \(u\sim p_u\) be a latent random vector and
\begin{equation}
    y=g(u),
\end{equation}
where \(g:\mathbb{R}^n\rightarrow\mathbb{R}^n\) is bijective and differentiable. 
The density of \(y\) is given by the change-of-variables formula
\begin{equation}\label{eq:change-of-var}
    p_y(y)=p_u\!\left(g^{-1}(y)\right)
    \left|
    \det J_{g^{-1}}(y)
    \right|.
\end{equation}
where $J_{f}$ denotes the Jacobian matrix of the mapping
$f$. Thus, normalizing flows extend invertible linear transformations to nonlinear vector mappings while accounting for the local volume change induced by the transformation. They have been applied to density estimation \cite{dinh2014nice}, image generation \cite{kingma2018glow}, audio generation \cite{prenger2019waveglow}, and motion generation \cite{luo2024potential}.

Beyond finite-dimensional random vectors, many applications require distributions over temporal signals and trajectories. Flow-based models have therefore been extended to probabilistic time-series forecasting \cite{rasulmultivariate}, continuous stochastic processes \cite{deng2020modeling}, conditional motion generation \cite{henter2020moglow}, and stable latent dynamic models \cite{urain2020imitationflow}. These works demonstrate the use of invertible models for temporal and sequential data, although the invertible transformation is typically defined on a finite-dimensional representation of a trajectory. In contrast, we propose to use a robustly invertible nonlinear dynamical system as the invertible mapping itself. The proposed BiLipREN maps a latent input signal to an output trajectory through nonlinear recurrent dynamics, and can therefore be applied to arbitrary-length signals and connects the two preceding viewpoints: Without internal dynamics, BiLipREN reduces to a static invertible map as in a normalizing flow, restricting to linear dynamics, it reduces to a minimum-phase system as in spectral factorization.

\subsection{Invertibility of nonlinear dynamical systems}
The nonlinear signal-space viewpoint above is closely related to the classical study of nonlinear system inversion. For nonlinear systems, the concept of minimum-phase is typically defined through the zero dynamics, i.e. internal dynamics of the system when the output remains identically zero. A nonlinear system is minimum-phase if its zero dynamics are stable \cite{byrnes1988local}. For non-minimum-phase systems the inverse is unstable, but factorizations based on the classical linear minimum-phase/all-pass a.k.a. inner-outer factorizations can be developed \cite{ball1992inner, van2018l2}, and approximate inversion can be achieved e.g. via finite horizon non-causal preview \cite{zou2007precision} or pseudo-inverse methods \cite{romagnoli2019general}. 

In much of the literature, invertibility of nonlinear systems is studied through a \emph{normal form} representation
\cite{hirschorn1979invertibility, devasia1996nonlinear}, i.e. a coordinate transformation that separates the system into internal dynamics and an input--output component, so that zero dynamics can be directly studied \cite{liberzon2002output}.
However, the normal form can be difficult to construct for general nonlinear systems.

Furthermore, while many existing approaches study stability of the inversion through asymptotically stable zero dynamics or input-to-state stability (ISS), they do not directly quantify the input-output gain of the inverse map from output signal perturbations to input reconstruction errors. To our knowledge,  the only works in this direction are \cite{liberzon2004output, liberzon2002output} which introduced the concept of output–input stability, however this requires a normal form construction and does not imply incremental stability as studied in this paper. 

\subsection{Contributions}
    \begin{itemize}
        \item  We introduce a definition of robust invertibility of nonlinear dynamics in terms of contraction and bi-Lipschitzness, i.e. inversion is robust to disturbances and initial condition errors. 
        \item We explicitly describe a bi-Lipschitz dynamical system using strongly input-output monotonicity (a special case of bi-Lipschitzness). By composing such models in layers with static/dynamic linear orthogonal layers, we obtain more general bi-Lipschitz model representations. 
        \item We direct parameterize such robustly invertible model by a bi-Lipschitz recurrent equilibrium network (BiLipREN), which admits an analytical inverse that is also a robustly invertible BiLipREN.
        \item We apply the proposed model to several problems in internal model control, surrogate loss learning, and generative modeling, showing the utility of robustly invertible dynamic models and illustrating the effect of various parameters.
    \end{itemize}
A preliminary paper related to this work appeared in \cite{zhang2025robustly}. The present paper provides new theoretical developments and complete proofs omitted from the conference version, together with new applications in internal model control, surrogate loss learning and generative modelling.

{\bf Notation.} Let $\mathbb{N}$ and $\mathbb{R}$ be the set of natural and real numbers, respectively. Let $\ell^n$ be the set of signals from $\mathbb{N}$ to $\R^n$. We use lower and upper case Latin letters such as $x$ and $A$ to denote vectors and matrices, respectively, and lower and upper case boldface Latin letters such as $\xb$ and $\Gb$ to denote signals and transfer matrices, respectively. We use calligraphic letters such as $\G$ to denote (nonlinear) operators between signal spaces. Let $\ell_2^n\subset \ell^n$ be the set of signals with finite $\ell_2$ norm, i.e., $\|\xb\|=\sqrt{\sum_{t=0}^{\infty} |x_t|^2}<\infty$ where $|\cdot|$ is the Euclidean norm. The notation $\langle \cdot,\cdot\rangle$ denotes the inner product $\langle x,y\rangle=x^\top y$ for $x,y\in\mathbb{R}^n$ and $\langle \xb,\yb\rangle=\sum_{t=0}^{\infty}x_t^\top y_t$ for $\xb,\yb\in \ell_2^n$. We use $\xb_{[T]}$ to denote the truncated sequence until time $T\in \mathbb{N}$, i.e., $\xb_{[T]}=(x_0,x_1,\ldots,x_T)$. We then denote the truncated norm and inner product by  $\|\xb\|_T=\|\xb_{[T]}\|$ and $\langle \xb,\yb\rangle_T=\langle \xb_{[T]},\yb_{[T]}\rangle$, respectively. We write $A(\succeq) \succ 0$ for positive (semi-)definite matrices. We use $\mathbb{D}^+$ to denote the set of positive definite diagonal matrices. 

\section{Problem Formulation}

We first recall the standard definition of invertibility of an operator acting on a vector space $\X$:

\begin{defn}[Invertibility]\label{defn:inv}
    An operator $\G:\X\mapsto\X$ is  said to be  \emph{invertible} if there exists an operator $\G^{-1}:\X\mapsto\X$ such that
    \begin{subequations}
        \begin{align}
        \G^{-1}(\G(\ub))=\ub\quad \forall \ub \in \X, \label{eq:left-inv}\\
        \G(\G^{-1}(\yb))=\yb\quad \forall \yb\in \X. \label{eq:right-inv}
    \end{align}
    \end{subequations}
    It is left or right invertible if \eqref{eq:left-inv} or \eqref{eq:right-inv} holds, respectively.
\end{defn}

In this paper, we consider signal to signal operators induced by nonlinear state space models of the form:
\begin{equation} \label{eq:system}
    x_{t+1}=f(x_t,u_t),\quad y_t=h(x_t,u_t),
\end{equation}
where $x_t\in\mathbb{R}^n$ is the state, and $u_t,y_t\in\mathbb{R}^m$ are the input and output, respectively. Here $f,h$ are locally Lipschitz functions. 

If the output map $h$ is invertible \textit{w.r.t.} the input $u$, then system \eqref{eq:system} admits a causal inverse by swapping the roles of $u$ and $y$. Specifically, the inverse system can be written as
\begin{equation}\label{eq:system-inv}
    x_{t+1}=f\bigl(x_t,h^{-1}(x_t,y_t)\bigr),\quad
    u_t=h^{-1}(x_t,y_t),
\end{equation}
where $ h^{-1}:\mathbb{R}^n\times\mathbb{R}^m\to\mathbb{R}^m$ satisfies
\begin{equation}
h\bigl(x,h^{-1}(x,y)\bigr)=y,
\quad
h^{-1}\bigl(x,h(x,u)\bigr)=u
\end{equation}
for all $x\in \R^n$ and $u,y\in \R^m$.   

We note that while such a state space model is invertible in the sense of \eqref{defn:inv}, its inverse may be unstable or highly sensitive to initial conditions or input/output perturbations. We are interested in characterizing invertible systems where bounded initial state mismatch and external perturbations lead to bounded input reconstruction errors defined by 
\begin{subequations}
    \begin{align}
        \bm{e}_u&=\G_b^{-1}(\G_a(\ub)+\Delta \yb)-\ub, \\
        \bm{e}_y&=\G_a(\G_b^{-1}(\yb)+\Delta \ub)-\yb,
    \end{align}
\end{subequations}
where $\G_a:\ub\mapsto \yb$ and $\G_b^{-1}:\yb\mapsto \ub$ denote the operators defined by system \eqref{eq:system} with $x_0=a$ and system \eqref{eq:system-inv} with $x_0=b$, respectively, see the schematic plot in Figure~\ref{fig:robust-invertible}. We formally define the notion of robust invertibility below, which reduces to Definition \ref{defn:inv} if $a=b$ and $\Delta \ub=\Delta\yb=0$.

\begin{defn}\label{def:r-invert}
    System \eqref{eq:system} is said to be \emph{robustly invertible} if it has a causal inverse \eqref{eq:system-inv}, and there exist $\lambda_{xu},\lambda_{xy}, \lambda_{yu},\lambda_{uy}\geq 0$ such that
    \begin{subequations}\label{eq:robust-invert}
    \begin{align}
       \|\bm{e}_u\|_T &\leq \lambda_{xu}|a-b|+\lambda_{yu}\|\Delta \yb\|_T\label{eq:robust-invert-a}\\
       \|\bm{e}_y\|_T &\leq\lambda_{xy}|a-b|+\lambda_{uy}\|\Delta \ub\|_T\label{eq:robust-invert-b}
   \end{align}
   \end{subequations}
   for all $\ub, \yb, \Delta\ub,\Delta\yb\in \ell^m$, initial states $a,b\in \R^n$ and $T\in \mathbb{N}$.
\end{defn}

\begin{figure}[!tb]
    \centering
    \begin{tabular}{c}
    \includegraphics[width=0.8\linewidth]{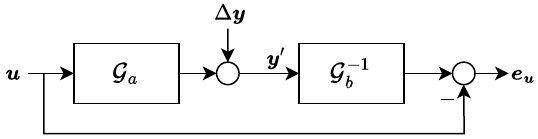}  \\
    (a) Input reconstruction for system \eqref{eq:system} via \eqref{eq:system-inv} \\
    \includegraphics[width=0.8\linewidth]{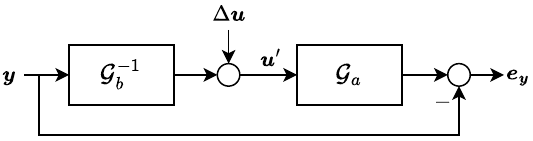} \\
    (b) Input reconstruction for system \eqref{eq:system-inv} via \eqref{eq:system}
    \end{tabular}
    \caption{Input reconstruction of system \eqref{eq:system} and its inverse \eqref{eq:system-inv} in the presence of initial state mismatch ($a\neq b$) and disturbances ($\Delta \ub, \Delta \yb\neq 0$).}
    \label{fig:robust-invertible}
\end{figure}

Our first technical problem is to characterize the conditions for robust invertibility. 
\begin{prob}
Provide tractable sufficient conditions for a  system of the form \eqref{eq:system} to be robustly invertible.
\end{prob}

\subsection{Learning Robustly Invertible State-Space Models}
We are also interested in the problem of learning robustly invertible models from a dataset $\mathfrak{D}$. Let $\mathfrak{M}:\theta \rightarrow (f_\theta,g_\theta)$ be a model parameterization, where for each $\theta \in \R^N$ we can obtain a state-space model \eqref{eq:system} with $f=f_\theta$ and $h=h_\theta$. The learning problem can be formulated as 
\begin{equation}\label{eq:learning}
    \min_{\theta \in \R^N}\; L(\theta,\mathfrak{D})
\end{equation}
where $L$ is the loss function. In the context of system identification we may have $\mathfrak{D}=(\tilde{\ub},\tilde{\yb})$ consisting of $T$-length input-output sequences and \emph{simulation error} as the loss function:
\begin{equation}
    L(\theta,\mathfrak{D})=\|\G_a(\tilde{\ub};\theta)-\tilde{\yb}\|_T^2
\end{equation}
where $\G_a(\cdot;\theta):\ub\mapsto \yb$ denotes the input-output mapping of system \eqref{eq:system} with $(f_\theta, h_\theta)$ and initial state $x_0=a$.

Another goal of this paper is a differential parameterization of robustly invertible models. 
\begin{prob}\label{prob}
    Construct a model parameterization $\mathfrak{M}:\theta\rightarrow (f_\theta, h_\theta)$ that satisfies the following requirements.
    \begin{enumerate}
    \item[\textbf{R1}] \textbf{Robust invertibility by design:} System \eqref{eq:system} defined by $(f_\theta,h_\theta)$ is robustly invertible for any $\theta\in \R^N$. 
    \item[\textbf{R2}] \textbf{Expressive:} Subject to Requirement \textbf{R1}, the model class is as flexible as possible to maximize coverage of high-performance models.
    \item[\textbf{R3}] \textbf{Differentiable:} The model parameterization $\mathfrak{M}$ is differentiable for all $\theta \in \mathbb R^N$.
 \end{enumerate}
\end{prob}
Requirements \textbf{R1} and \textbf{R3} are hard requirements. \textbf{R1} allows us to decouple robust invertibility from the choice of loss functions, datasets and optimization algorithm for the learning problem \eqref{eq:learning}. \textbf{R3} ensures that the model class is compatible with standard machine learning tools based on automatic differentiation. \textbf{R2} is a soft requirement  that the robust invertibility constraints are not overly restrictive, so high-performing models can be learned for diverse applications. 

\section{Background}

\subsection{Causal and Bi-Lipschitz Operators}
We start with some standard definitions for operators.
\begin{defn}[Causality]
An operator $\G:\ell^m\mapsto\ell^n$ is said to be \emph{causal} if for any pair of input sequences $\ub,\ub'\in \ell^m$,
\begin{equation}
    \ub'_{[T]}=\ub_{[T]}\;\Longrightarrow\; \G(\ub')_{[T]}=\G(\ub)_{[T]}
\end{equation}
for all $T\in \mathbb{N}$.
\end{defn}

\begin{defn}[Lipschitzness]
    An operator $\G:\ell^m\mapsto\ell^n$ is said to be \emph{$\nu$-Lipschitz} with $\nu>0$ if for any pair of input sequences $\ub,\ub'\in \ell^m$
    \begin{equation}\label{eq:lipschitz}
        \|\G(\ub')-\G(\ub)\|_T\leq \nu \|\ub'-\ub\|_T
    \end{equation}
   for all $T\in \mathbb{N}$. It is \emph{$\mu$-inverse Lipschitz} with $\mu>0$ if
    \begin{equation}\label{eq:inv-lipschitz}
        \|\G(\ub')-\G(\ub)\|_T\geq \mu \|\ub'-\ub\|_T.
    \end{equation}
   The operator $\G$ is said to be \emph{bi-Lipschitz} with $\nu\geq \mu>0$, or simply $(\mu,\nu)$ bi-Lipschitz, if both \eqref{eq:lipschitz} and \eqref{eq:inv-lipschitz} hold. 
\end{defn}

\begin{remark}
   Note that the \emph{composition rule} holds for invertibility, causality and bi-Lipschitz continuity. If $\G_1,\G_2:\ell^m\rightarrow\ell^m$ are invertible, then their composition  $\G=\G_2\circ \G_1$ is also invertible, with  inverse $\G^{-1}=\G_1^{-1}\circ\G_2^{-1}$. Similarly, if both $\G_1$ and $\G_2$ are casual, then $\G$ is causal. If $\G_i$ is $(\mu_i,\nu_i)$ bi-Lipschitz for $i=1,2$, then $\G$ is $(\mu,\nu)$ bi-Lipschitz with $\mu=\mu_1\mu_2$ and $\nu=\nu_1\nu_2$.  
\end{remark}

\begin{remark}
     Any finite-dimensional bi-Lipschitz map $g:\R^m\rightarrow\R^m$ has a well-defined bi-Lipschitz inverse $g^{-1}$ \cite{wangmonotone}. However, this generally does not hold for bi-Lipschitz operators $\G:\ell^m\rightarrow\ell^m$. For example, we consider the right shift operator
    \[
    \mathcal{S}_R:(u_0,u_1,\ldots)\mapsto (0, u_0,u_1,\ldots),
    \]
    which is bi-Lipschitz with $\mu=\nu=1$ but not invertible as it is not a surjective map. It does admit a left inverse, i.e., the left shift operator
    \[
    \mathcal{S}_L:(y_0,y_1,y_2\ldots)\mapsto (y_1,y_2,\ldots),
    \]
    However, this is neither causal nor bi-Lipschitz. 
\end{remark}

\subsection{Contracting and Lipschitz State-Space Models}

To characterize robust invertibility, we will use a strong notion of internal stability as follows.
\begin{defn}[Contraction]
    System \eqref{eq:system} is \emph{contracting} if for any two initial conditions $a,b\in\mathbb{R}^{n}$ and the same input signal $\ub\in \ell^m$, the corresponding state trajectories $\xb',\xb$ satisfy 
    \begin{equation}\label{eq:contracting}
        |x_t'-x_t|\leq \kappa \alpha^t |x_0'-x_0| \quad \forall t\in \mathbb{N}
    \end{equation}
    with overshoot $\kappa \geq 1$ and contraction rate $\alpha \in [0,1)$.
\end{defn}
The composition of two contracting systems is still contracting \cite{lohmiller1998contraction}. Beyond internal stability, we also need the input-output robustness described by incremental integral quadratic constraint ($\delta$IQCs) \cite{megretski1997system}.

\begin{defn}[$\delta$IQC]\label{dfn:IQC}
    System \eqref{eq:system} is said to satisfy the  \emph{incremental integral quadratic constraint} ($\delta$IQC) defined by a quadratic function $s(\Delta u, \Delta y)$, if any two trajectories $(\xb',\ub',\yb'),(\xb,\ub,\yb)$ satisfy
    \begin{equation} \label{eq:iqc}
        \sum_{t=0}^{T} 
        s(\Delta u_t,\Delta y_t)
        \ge -\kappa(x_0', x_0)
        \quad  \forall T\in\mathbb{N}
    \end{equation}
    with $\Delta y_t = y_t' - y_t$ and $\Delta u_t = u_t' - u_t$, where $\kappa(x_0', x_0) \ge 0$ and $\kappa(x_0, x_0) = 0$ for all $x_0',x_0\in\R^n$.
\end{defn}

System \eqref{eq:system} is said to be \emph{$\nu$-Lipschitz} (respectively, \emph{$\mu$-inverse Lipschitz}) if it satisfies the $\delta$IQC with $s(\Delta u,\Delta y)=\nu^2|\Delta u|^2-|\Delta y|^2$ (respectively $s(\Delta u,\Delta y)=|\Delta y|^2-\mu^2|\Delta u|^2$). And we call \eqref{eq:system} a $(\mu,\nu)$ bi-Lipschitz system with $\nu\geq \mu>0$ if it is both $\mu$-inverse Lipschitz and $\nu$-Lipschitz. 

The contraction property of system \eqref{eq:system} can be certified by an \emph{incremental Lyapunov function} function $V:\R^n\times\R^n\rightarrow \R_+$ which satisfies
\begin{gather}
    c_1|x'-x|^2\leq V(x',x)\leq c_2|x'-x|^2\quad \forall x',x\in \R^n \label{eq:uniform-bound}\\
    V(x'_{t+1},x_{t+1})-\alpha^2V(x'_t,x_t)\leq 0
\end{gather}
for some $c_2\geq c_1>0$ and $\alpha \in [0,1)$, where $\xb',\xb$ are any two state trajectories of \eqref{eq:system} under the same input $\ub$. Similarly, we can obtain \eqref{eq:iqc} via the following incremental dissipation inequality: 
\begin{equation}
    V(x_{t+1}',x_{t+1})-V(x_t',x_t)\leq s(\Delta u_t,\Delta y_t)
\end{equation}
where $V:\R^n\times\R^n\rightarrow \R_+$ is an \emph{incremental storage function} with $V(x',x)\geq 0$ and $V(x,x)=0$ for all $x',x\in \R^n$.

\section{Robustly Invertible Bi-Lipschitz Systems via Strongly Monotone and Orthogonal Layers}

In this section, we first present sufficient conditions for robust invertibility based on contraction and Lipschitz properties. We then develop a systematic approach for constructing robustly invertible state-space models through the composition simpler layers which satisfy the robust invertibility properties. Finally, we introduce a nonlinear inner-outer factorization incorporating dynamic orthogonal (all-pass/inner) layers.. 

\subsection{Robust Invertible Bi-Lipschitz Systems}

We first relate robust invertibility to contraction and Lipschitzness of system \eqref{eq:system} and its inverse \eqref{eq:system-inv}.
\begin{thm}\label{thm:robust-inv}
    Suppose that system \eqref{eq:system} is $\nu$-Lipschitz and contracting with rate $\alpha$ and overshoot $\kappa$, and its causal inverse \eqref{eq:system-inv} is $1/\mu$-Lipschitz with $\mu\le \nu$ and contracting with rate $\hat{\alpha}$ and overshoot $\hat{\kappa}$. If the output mappings $h, h^{-1}$ have Lipschitz bounds of $\gamma$ and $\hat{\gamma}$ w.r.t. $x$, respectively, then  system \eqref{eq:system} is robustly invertible with 
    \begin{equation}\label{eq:bounds}
        \begin{split}
            \lambda_{xu}=\frac{\hat{\gamma}\hat{\kappa}}{\sqrt{1-\hat{\alpha}^2}}, \quad \lambda_{xy}=\frac{\gamma\kappa}{\sqrt{1-\alpha^2}}, \quad \lambda_{yu}=\frac{1}{\mu},\quad\lambda_{uy}=\nu.
        \end{split}
    \end{equation} 
\end{thm}
\begin{remark}
    If a system is $\nu$-Lipschitz and it has well-defined inverse which is $1/\mu$-Lipschitz, then that system is $(\mu,\nu)$ bi-Lipschitz. 
\end{remark}
\begin{proof}
    By letting $\yb=\G_a(\ub)$, we have
    \begin{equation}\label{eq:bound-eu}
        \begin{split}
            \|\bm{e}_u\|_T=&\|\G_b^{-1}(\yb+\Delta \yb)-\ub\|_T \\
            \le & \|\G_b^{-1}(\yb+\Delta \yb)-\G_b^{-1}(\yb)\|+\|\G_b^{-1}(\yb)-\ub\|_T \\
            \le & \frac{1}{\mu}\|\Delta \yb\|_T+\|\G_b^{-1}(\yb)-\ub\|_T \\
            =& \frac{1}{\mu}\|\Delta \yb\|_T+\|\G_b^{-1}(\yb)-\G_a^{-1}(\yb)\|_T
        \end{split}
    \end{equation}
    where the second inequality follows as \eqref{eq:system-inv} is $1/\mu$-Lipschitz. Moreover, since it is also contracting, we can obtain
    \begin{equation}\label{eq:bound-ex}
        \begin{split}
            \|\G_b^{-1}(\yb)&-\G_a^{-1}(\yb)\|_T \leq \hat{\gamma}\|\xb^b-\xb^a\|_T \\
            &\leq \sqrt{\sum_{t=0}^T(\hat{\kappa}\hat{\alpha}^{t}|a-b|)^2}\leq \frac{\hat{\kappa}}{\sqrt{1-\hat{\alpha}^2}}|a-b|
        \end{split}
    \end{equation}
    where $\xb^a,\xb^b$ are the state trajectories of $\Gb_a^{-1}(\yb),\Gb_b^{-1}(\yb)$, respectively. Combining \eqref{eq:bound-eu} and \eqref{eq:bound-ex} yields the input reconstruction error bound \eqref{eq:robust-invert-a} for \eqref{eq:system}. We can then establish the bound \eqref{eq:robust-invert-b} for \eqref{eq:system-inv} by the same procedure.
\end{proof}

\begin{corollary}\label{coro:1}
    Let $\G_1,\G_2, ..., \G_K$ be dynamical systems with state-space realization of \eqref{eq:system}. If $\G_i$ are invertible for all $i=1, ..., K$, then 
    \begin{equation}
\G=\G_K\circ\cdots\circ\G_2\circ\G_1,
\end{equation}
is also robustly invertible.
\end{corollary}
\begin{proof}
    Since invertibility, causality, contraction, and Lipschitz continuity are all preserved under composition, robust invertibility is likewise preserved under composition.
\end{proof}
The above corollary enables the construction of deep robustly invertible models by composing simple robustly invertible layers. In Sections~\ref{sec:siom} and \ref{sec:dynamic-orthogonal} below, we introduce two such building blocks: a contracting and strongly I/O monotone dynamical layer and a static orthogonal layer \cite{trockman2021orthogonalizing}.

\subsection{Contracting and Strongly I/O Monotone Systems}\label{sec:siom}

We first introduce the concept of strongly input-output monotone operators, similar to the classical notion of incremental passivity \cite{desoer2009feedback}. 
\begin{defn}\label{dfn:monotone}
    An operator $\M:\ell^m\mapsto \ell^m$ is said to be \emph{strongly input-output monotone} with $\sigma, \eta>0$ and $\sigma\eta\leq 1$, or simply $(\sigma,\eta)$-strongly I/O monotone, if
    \begin{equation}\label{eq:in_mono}
       2 \langle \Delta \yb , \Delta \ub \rangle_T \geq \sigma\|\Delta \ub\|^2_T+\eta\|\Delta \yb \|^2_T\quad 
    \end{equation}
    for all $\ub,\ub'\in\ell^m$ and $T\in \mathbb{N}$, where $\Delta \yb =\M(\ub')-\M(\ub)$ and $\Delta \ub = \ub'-\ub$. 
\end{defn}
\begin{remark}
Note that $\sigma\eta\leq 1$ is necessary for \eqref{eq:in_mono} since 
\[
2\|\Delta \yb\|_T\|\Delta \ub\|_T\geq 2 \langle \Delta \yb , \Delta \ub \rangle_T\geq 2\sqrt{\sigma\eta} \|\Delta \yb\|_T\|\Delta \ub\|_T.
\]
\end{remark}

\begin{lem}\label{lem:iqc}
    Let $\M$ be a ($\sigma, \eta$)-strongly I/O monotone operator.
    \begin{enumerate}
    \item $\M$ is $(\mu,\nu)$-Lipschitz with
    \begin{equation}\label{eq:uv}
        \mu=\frac{1-\sqrt{1-\sigma\eta}}{\eta},\quad \nu=\frac{1+\sqrt{1-\sigma\eta}}{\eta}.
    \end{equation}
        \item $\M$ is invertible and $\M^{-1}$ is $(\eta,\sigma)$-strongly I/O monotone.
    \end{enumerate}
\end{lem}
\begin{proof}
Claim 1): From \eqref{eq:uv} we have the relationship 
\begin{equation}\label{eq:monotone-bound}
    \sigma = \frac{2\mu\nu}{\mu+\nu}, \quad \eta = \frac{2}{\mu+\nu}.
\end{equation}
By applying Cauchy-Schwarz inequality to  \eqref{eq:in_mono} we obtain 
\begin{equation*}
    \|\Delta \yb\|_T\|\Delta \ub\|_T\geq  \langle \Delta \yb , \Delta \ub \rangle_T\geq \frac{1}{\mu+\nu} (\mu\nu\|\Delta \ub\|_T^2+\|\Delta \yb\|_T^2)
\end{equation*}
which further implies
\begin{equation}\label{eq:ginv-bilip}
    \left(\|\Delta \yb\|_T-\mu \|\Delta \ub\|_T\right)\left(\nu\|\Delta \ub\|_T-\|\Delta \yb\|_T\right)\geq 0.
\end{equation}
Thus, $\M$ is $(\mu,\nu)$ bi-Lipschitz. 

Claim 2): Since $\M$ is $\mu$ inverse Lipschitz, then we have 
\[
\Delta \yb=0 \; \Longrightarrow\; \|\Delta \ub\|_T \leq \frac{1}{\mu}\|\Delta \yb\|_T=0,
\]
i.e., $\M$ is a injective mapping. Moreover, \eqref{eq:in_mono} implies
\begin{equation}\label{eq:monotone}
    \langle \Delta \yb , \Delta \ub \rangle_T \geq \frac{\sigma}{2}\|\Delta \ub\|^2_T.
\end{equation}
By Browder-Minty theorem \cite{zeidler1986nonlinearI}, $\M$ is also a surjective mapping. Thus, $\M$ is invertible. By swapping the roles of $u$ and $y$ in \eqref{eq:in_mono}, we can obtain that $\M^{-1}$ is $(\eta,\sigma)$-strongly I/O monotone. 
\end{proof}

\begin{remark}
Strongly I/O monotone operators constitute a well-behaved subclass of bi-Lipschitz operators: they are invertible, and their inverses are also strongly I/O monotone. However, strongly I/O monotonicity is not preserved under composition. Given two strongly I/O monotone $\M_1, \M_2$, the composition $\M=\M_2\circ\M_1$ does not need to be strongly I/O  monotone. Nevertheless, $\M$ remains bi-Lipschitz and invertible. 
\end{remark}

\begin{figure}[!bt]
\centering
    \includegraphics[width=0.7\linewidth]{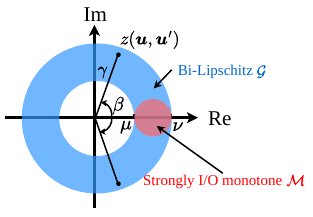}
    \caption{SRG of bi-Lipschitz and strongly input-output monotone operators with the same $\mu,\nu$ parameters. }\label{fig:srg} 
\end{figure}

We can visualize the relationship with a scaled relative graph (SRG), a graphical tool for analyzing operators \cite{ryu2022scaled}. Given an operator $\G$, its SRG is defined as a set of complex numbers $z(\ub,\ub')=\gamma e^{\pm j\beta}$ with 
\begin{equation}
\gamma =\frac{\|\Delta \yb\|}{\|\Delta \ub\|},\quad \beta = \arccos\frac {\langle \Delta\ub , \Delta\yb \rangle}{\|\Delta\ub\| \|\Delta\yb\|},
\end{equation} 
where $\Delta \yb =\G(\ub')-\G(\ub)$ and $\Delta \ub = \ub'-\ub \in \ell_2^m$. For dynamical systems, the SRG can be viewed as a generalization of the classical Nyquist plot of a linear system to much broader classes of nonlinear systems \cite{chaffey2023graphical}. Simple algebraic manipulation of \eqref{eq:in_mono} yields that the SRG of a $(\sigma,\eta)$ strongly I/O monotone operator $\M$ is a disk described by
\begin{equation}\label{eq:monotone-srg}
    \left|z(\ub,\ub')-\frac{\mu+\nu}{2} \right|\leq \frac{\nu-\mu}{2}
\end{equation}
where $\mu,\nu$ are given by \eqref{eq:uv}. Hence, it is a convex subset of the SRG of a $(\mu,\nu)$ bi-Lipschitz operator $\G$, see Figure~\ref{fig:srg}. 

We now focus on the operators realized by the state-space model \eqref{eq:system} which satisfy both strong I/O monotonicity property and internal stability of the system, characterized in terms of contraction.

\begin{thm}\label{thm:monotone-system}
    Consider the nonlinear state-space model \eqref{eq:system}. If there exists a uniformly-bounded function $V(x',x)$ such that
    \begin{equation}\label{eq:strict-disspative}
    V(x'_{t+1},x_{t+1})-\alpha^2V(x'_t,x_t)\leq s(\Delta u_t,\Delta y_t)
    \end{equation}
    for some $\alpha \in [0, 1)$, where 
    \begin{equation}\label{eq:io-supply}
        s(\Delta u,\Delta y)=2\langle \Delta y,\Delta u\rangle-\sigma |\Delta u|^2-\eta |\Delta y|^2
    \end{equation}
    with $\sigma\eta\leq 1$ and $\sigma,\eta>0$, then \eqref{eq:system} is contracting and $(\sigma,\eta) $-strongly I/O monotone. Moreover, it has a causal inverse \eqref{eq:system-inv}, which is  contracting and $(\eta,\sigma)$ strongly I/O monotone.
\end{thm}
\begin{remark}
The benefit of the above condition is that a \textit{single} dissipation inequality \eqref{eq:strict-disspative} simultaneously verifies that the forward model \eqref{eq:system} is contracting and bi-Lipschitz \textit{and} that the inverse model \eqref{eq:system-inv} is well-posed, contracting, and bi-Lipschitz. This will enable a differentiable parameterization of robustly invertible models in the form of a REN \cite{revay2024recurrent} in Section \ref{sec:bilipren}.     
\end{remark}
\begin{proof}
First, by taking $\Delta u_t\equiv 0$, we obtain 
\[
V(x'_{t+1},x_{t+1})-\alpha^2V(x'_t,x_t)\leq -\eta |\Delta y_t|^2\leq 0
\]
which implies that system \eqref{eq:system} is contracting. Eq.~\eqref{eq:strict-disspative} implies that system \eqref{eq:system} satisfies the $\delta$IQC defined by \eqref{eq:io-supply}, i.e., it is $(\sigma,\eta)$ strongly I/O monotone.

Second, we show that the output mapping $h$ is invertible {\it w.r.t.} the system input $u$. For any $a\in \R^n$, we take $x_0'=x_0=a$ for \eqref{eq:strict-disspative}, which leads to 
\begin{equation}
    2\langle \Delta y,\Delta u\rangle-\sigma |\Delta u|^2-\eta |\Delta y|^2 \geq V(x_1',x_1)\geq 0
\end{equation}
where $\Delta u=u_0'-u_0$ and $\Delta y=h(a,u_0')-h(a,u_0)$. By following the proof of Lemma~\ref{lem:iqc} we can conclude that $h(a,\cdot):\R^m\rightarrow \R^m$ is $(\mu,\nu)$ bi-Lipschitz for $a\in \R^n$, implying that $h$ is invertible {\it w.r.t.} $u$. Therefore, a causal inverse of system \eqref{eq:system} can be obtained via \eqref{eq:system-inv}.

Finally, we show the stability properties of \eqref{eq:system-inv}. 
Since \eqref{eq:system} and \eqref{eq:system-inv} share the same set of solutions $(\xb,\ub,\yb)$, the incremental dissipative inequality \eqref{eq:strict-disspative} holds for \eqref{eq:system-inv} with $y$ as the input and $u$ as the output. Similarly, we have that system \eqref{eq:system-inv} is contracting and $(\eta,\sigma)$-strongly I/O monotone.
\end{proof}

\subsection{Static and Dynamic Orthogonal Layers and Nonlinear Inner-Outer Factorization}\label{sec:dynamic-orthogonal}

Another basic layer with the robust invertibility property is the \textit{static orthogonal layer} $\mathcal{O}:\R^m\rightarrow\R^m$ \cite{trockman2021orthogonalizing}, i.e. an affine map of the form
\begin{equation}
    \mathcal{O}(u)=P u+q
\end{equation}
where $P\in \R^{m\times m}$ is an orthogonal weight matrix ($P^\top P=PP^\top=I$), $q\in \R^m$ is a bias term. It is easy to verify that $\Oc$ is bi-Lipschitz with $\mu=\nu=1$ and its inverse is also an orthogonal layer:
\begin{equation}
    \Oc^{-1}(y) = P^\top(y-q).
\end{equation}

We also introduce the \textit{dynamic orthogonal layer}, which is not robustly invertible in the above sense but is bi-Lipschitz and is stable in reverse time. A dynamic orthogonal layer is an affine system $\bm{O}: \ell^m\mapsto \ell^m $ of the form
\begin{equation}\label{eq:lti-orth}
    \Ob:\;\left\{\begin{aligned}
        x_{t+1} &= Ax_t+Bu_t+b_x\\
        y_t &=Cx_t+Du_t+b_y
    \end{aligned}\right.
\end{equation}
with $x_t\in \R^n$ and $ u_t,y_t\in \R^m$, where 
\begin{equation}
Q=\begin{bmatrix}
A & B \\ C & D
\end{bmatrix}
\end{equation}
is orthogonal, and $b_x\in \R^n,b_y\in \R^m$ are bias terms. If we omit bias terms, the associated transfer function from $\Delta \ub$ to $\Delta \yb$ is $\Ob(z)= C(zI-A)^{-1}B+D$, which is an \emph{all-pass filter} (see e.g. \cite{OppenheimSchafer2010, heuberger2005modelling}) since
\begin{equation}\label{eq:orth-dc}
    \Ob\bigl(z^{-1}\bigr)^\top \Ob(z)=I.
\end{equation}
Furthermore, since $Q$ is an orthogonal matrix, the system satisfies the equality $|x_{t+1}|^2 - |x_t|^2 = |u_t|^2 - |y_t|^2$, from which it follows that \eqref{eq:lti-orth} is stable and $(1,1)$ bi-Lipschitz, and also that its inverse is unstable, whereas its anti-causal inverse is stable in reverse time:
\cite{heuberger2005modelling}:
\begin{equation}\label{eq:dyn-orth-inv}
\begin{bmatrix}
    x_t \\ u_t
\end{bmatrix}=
\begin{bmatrix}
    A & B \\ C & D
\end{bmatrix}^\top 
\begin{bmatrix}
    x_{t+1}-b_x \\ y_t -b_y
\end{bmatrix}.
\end{equation}
Thus, the dynamic orthogonal layer can also be utilized to construct models which have a non-causal robust inverse over fixed-length (batch) data.

The composition of an all-pass filter and a robustly invertible nonlinear system can be considered as a form of inner-outer factorization of a stable nonlinear system \cite{ball1992inner}, generalizing the classical minimum-phase/all-pass factorization of stable linear systems.

In classical inner-outer factorization, a model is given and the inner and outer factors are computed. In this paper, we propose to learn models from data in an ``already factored'' form:
\begin{equation}\label{eq:inner-outer}
    \mathcal{P}=\Ob(z)\circ \G
\end{equation}
where $\Ob(z)$ is a dynamic orthogonal layer (the all-pass a.k.a. inner factor), $\G$ is a robustly invertible bi-Lipschitz system (the minimum-phase a.k.a. outer factor). 

Although $\mathcal{P}$ is contracting and bi-Lipschitz, its inverse $\mathcal{P}^{-1}$ is unstable. However, for batch data the robustly invertible layer $\mathcal G$ can be inverted forwards in time, and the orthogonal layer $\bm{O}(z)$ can be inverted in reverse time.  Furthermore, to recover the steady-state input, we introduce the approximate inverse:
\begin{equation}
    \mathcal{P}^{\sharp}=\G^{-1}\circ \Ob(1)^\top,
\end{equation}
where the reconstruction error $\bm{e}_u=\mathcal{P}\circ\mathcal{P}^{\sharp}(\ub)-\ub$  converges to zero exponentially for any constant input $u_t\equiv u$. 

\begin{remark}
In our proposed framework \eqref{eq:inner-outer}, the all-pass/inner factor is linear (or affine with bias terms), whereas in previous approaches to nonlinear inner-outer factorization the inner factor is energy-preserving ($\|\yb\|=\|\ub\|$) but may be nonlinear (e.g. \cite{ball1992inner, van2018l2}). It is reasonable to ask whether in the incremental setting the inner factor could also be nonlinear. However, the Mazur–Ulam theorem \cite{mazur1932transformations, nica2012mazur} states that on any real normed linear space (even infinite dimensional), every invertible distance-preserving map is affine. So in this sense, there is no loss of generality from restricting to affine models.
\end{remark}

\section{Parameterization of Bi-Lipschitz Recurrent Equilibrium Networks}\label{sec:bilipren}

In this section, we introduce a differentiable parameterization of the Bi-Lipschitz REN (\emph{BiLipREN}), which is a robustly invertible bi-Lipschitz state-space model in the form of a \emph{Recurrent Equilibrium Network} (REN) \cite{revay2024recurrent}. Furthermore, BiLipREN admits an analytical inverse which is also a BiLipREN.

\subsection{Recurrent Equilibrium Networks}
REN is a neural dynamic model which is formulated as a feedback interconnection of a linear time-invariant (LTI) system and a static activation:
\begin{subequations}\label{eq:ren}
\begin{align}
\renewcommand\arraystretch{1.2}
\begin{bmatrix}
x_{t+1} \\ v_t \\ y_t
\end{bmatrix}&=
\overset{W}{\overbrace{
		\left[
		\begin{array}{c|cc}
		A & B_1 & B_2 \\ \hline 
		C_{1} & D_{11} & D_{12} \\
		C_{2} & D_{21} & D_{22}
		\end{array} 
		\right]
}}
\begin{bmatrix}
x_t \\ w_t \\ u_t
\end{bmatrix}+
\overset{b}{\overbrace{
		\begin{bmatrix}
		b_x \\ b_v \\ b_y
		\end{bmatrix}
}}, \label{eq:G}\\
w_t=&\sigma(v_t):=\begin{bmatrix}
\sigma(v_{t}^1) & \sigma(v_{t}^2) & \cdots & \sigma(v_{t}^q)
\end{bmatrix}^\top, \label{eq:sigma}
\end{align}   
\end{subequations}
where $x_t\in \R^n, u_t,y_t\in \R^m$ are the state, input and output, respectively. Here $v_t, w_t\in \R^q$ are the input and output variables of an activation function $\sigma$ with slope-restricted in $[0,1]$. Here $(W,b)$ is the learnable parameter. 

RENs contain an algebraic loop in \eqref{eq:ren} which, if well-posed, yields an \emph{equilibrium network} $\phi:(x_t,u_t)\rightarrow w_t$, where $w_t$ is the solution of the following implicit equation
\begin{equation}\label{eq:implicit}
w_t=\sigma(D_{11} w_t+b_w),
\end{equation}
with $b_w=C_1x_t+D_{12}u_t+b_v$. Thus, REN \eqref{eq:ren} can be rewritten as \eqref{eq:system} with 
\[
\begin{split}
    f(x,u)&=Ax+B_1\phi(x,u)+B_2u+b_x, \\
    h(x,u)&=C_2x+D_{21}\phi(x,u) + D_{22}u+b_y.
\end{split}
\]
Note that the equilibrium network \eqref{eq:implicit} is not a single layer but include many existing deep feed-forward structures| e.g., multi-layer perception (MLP), residual network (ResNet) and convolutional neural network (CNN)| as special cases \cite{ghaoui2019implicit}. 

A central result of \cite{revay2020lipschitz} shows that if there exists a $\Lambda\in \mathbb{D}^+$ with $\mathbb{D}^+$ as the set of positive diagonal matrices such that 
\begin{equation}\label{eq:wellpose}
    2\Lambda-\Lambda  D_{11}- D_{11}^\top\Lambda\succ0,
\end{equation}
then the equilibrium network \eqref{eq:implicit} is well-posed, i.e., it admits a unique solution $w_t$ for any $b_w\in \R^q$.

\subsection{Strongly I/O Monotone REN}

We now state our third main theoretical result as follows. 
\begin{thm}\label{thm:main}
Consider the REN model \eqref{eq:ren}. Suppose that there exist $P = P^\top\succ 0$ and $\Lambda \in \mathbb{D}^+$ such that
\begin{align}
   &\begin{bmatrix}\nonumber
       P& -C_1^\top \Lambda&C_2^\top \\
       -\Lambda C_1 & \mathcal{W} &D_{21}^\top - \Lambda D_{12}\\
       C_2&D_{21}- D_{12}^\top\Lambda &-\eta I+D_{22}+D_{22}^\top.
   \end{bmatrix} \\
   &\; -
   \begin{bmatrix}
       A^\top\\
       B_1^\top\\
       B_2^\top 
   \end{bmatrix}P\begin{bmatrix}
       A^\top\\
       B_1^\top\\
       B_2^\top 
   \end{bmatrix}^\top-\sigma\begin{bmatrix}
       C_2^\top\\
       D_{21}^\top\\
       D_{22}^\top
   \end{bmatrix}\begin{bmatrix}
       C_2^\top\\
       D_{21}^\top\\
       D_{22}^\top
       \end{bmatrix}^\top\succ 0 \label{eq:passivity}
    \end{align}
with $\sigma\eta\leq 1$ and $\sigma,\eta>0$, where $\mathcal{W}=2\Lambda -\Lambda D_{11}-D_{11}^\top \Lambda$.
\begin{itemize}
    \item[1)] REN \eqref{eq:ren} is well-posed, invertible, contracting and $(\sigma,\eta)$-strongly I/O monotone.
    \item[2)] It has a causal inverse,  which is a well-posed, contracting and $(\eta,\sigma)$-strongly I/O monotone REN.  
\end{itemize}
\end{thm}

\begin{remark}
    By adopting the direct parameterization in \cite{revay2024recurrent} with supply rate matrices $Q = -\sigma I, S = I, R = -\eta I$, one obtains a differentiable parameterization $\mathcal{M}:\theta \rightarrow (W,b)$ for contracting and strongly I/O monotone RENs, i.e., \eqref{eq:passivity} holds for any $\theta \in \R^N$. Moreover, REN \eqref{eq:ren} is robustly invertible and $(\mu,\nu)$ bi-Lipschitz with $\mu,\nu$ given in \eqref{eq:uv}.
\end{remark}

\begin{remark}
    The direct parameterization in \cite{revay2024recurrent} also yields a useful subclass of RENs, called \emph{acyclic RENs}, in which the weight $D_{11}$ is guaranteed to be strictly lower triangular. As a result, the implicit layer \eqref{eq:implicit} admit a unique solution $w_t$ that can be computed explicitly by forward substitution, row by row. Although the inverse of an acyclic and strongly I/O monotone REN is still a well-posed REN \eqref{eq:iren}, its corresponding weight matrix $\hat{D}_{11}$ is generally a full matrix. Consequently, the implicit layer of the REN inverse model typically requires an iterative solver based on monotone operator splitting; see \cite{revay2020lipschitz} for details. Depending on the application, one may choose to parameterize either $\M$ or $\M^{-1}$ as an acyclic REN.
\end{remark}

\begin{proof}
Claim 1): First, the well-posedness of REN \eqref{eq:ren} follows directly from \eqref{eq:passivity}. Moreover, condition~\eqref{eq:passivity} implies that REN \eqref{eq:ren} satisfies \eqref{eq:strict-disspative} with $V(x,x')=(x'-x)^\top P(x'-x)$ and some $\alpha\in [0,1)$, see  
\cite[Thm.~1]{revay2024recurrent}. And Thm.~\ref{thm:monotone-system} shows that REN \eqref{eq:ren} is contracting and $(\sigma,\eta)$ strongly I/O monotone. 

Claim 2): From \eqref{eq:passivity} we have $D_{22}+D_{22}^\top \succeq \eta I$, i.e., $D_{22}$ is invertible.  We then construct a causal inverse of \eqref{eq:ren} by taking $u_t=D_{22}^{-1}(-C_2x_t-D_{21}w_t+y-b_y)$ and substituting it back to \eqref{eq:ren}. The REN inverse can also be written as a REN:
\begin{equation}\label{eq:iren}
    \begin{split}
    \renewcommand\arraystretch{1.2}
\begin{bmatrix}
x_{t+1} \\ v_t \\ u_t
\end{bmatrix}&=
\overset{\hat W}{\overbrace{
		\left[
		\begin{array}{c|cc}
		\hat A & \hat B_1 & \hat B_2 \\ \hline 
		\hat C_{1} & \hat D_{11} & \hat D_{12} \\
		\hat C_{2} & \hat D_{21} & \hat D_{22}
		\end{array} 
		\right]
}}
\begin{bmatrix}
x_t \\ w_t \\ y_t
\end{bmatrix}+
\overset{\hat b}{\overbrace{
		\begin{bmatrix}
		\hat b_x \\ \hat b_v \\ \hat b_y
		\end{bmatrix}
}} \\
        w_t&=\sigma(v_t)
    \end{split}
\end{equation}
where 
\begin{equation}\label{eq:irenparam}
    \begin{split}
        &\hat{A}=A-B_2D_{22}^{-1}C_2,\; \hat{B}_1 = B_1-B_2D_{22}^{-1}D_{21}, \\
        &\hat{B}_2 = B_2D_{22}^{-1},\; \hat{C}_1 = C_1-D_{12}D_{22}^{-1}C_2\\
        & \hat{C}_2=-D_{22}^{-1}C_2,\;
        \hat{D}_{11}=D_{11}-D_{12}D_{22}^{-1}D_{21},\\
        &\hat{D}_{12}=D_{12}D_{22}^{-1},\;\hat{D}_{21}=-D_{22}^{-1}D_{21},\; \hat{D}_{22}=D_{22}^{-1}\\
        &\hat{b}_x = b_x-B_2D_{22}^{-1}b_y,\; \hat{b}_v = b_v-D_{12}D_{22}^{-1}b_y,\\
        &\hat{b}_y=-D_{22}^{-1}b_y.
    \end{split}
\end{equation}
We can verify that the model parameters of \eqref{eq:iren} and \eqref{eq:ren} satisfying $\hat{W}= W\Psi$ and $\hat{b}= b\Psi$ with 
\begin{equation}\label{eq:Psi}
    \Psi=\begin{bmatrix}
        I & 0 & 0 \\
        0 & I & 0 \\
        -D_{22}^{-1}C_2 & -D_{22}^{-1}D_{21} & D_{22}^{-1}
    \end{bmatrix}.
\end{equation}
By left- and right-multiplying \eqref{eq:passivity} with $\Psi^\top$ and $\Psi$, respectively, we can obtain
\begin{align}
   &\begin{bmatrix}\nonumber
       P& -\hat C_1^\top \Lambda & \hat C_2^\top \\
       -\Lambda \hat C_1 & \hat{\mathcal{W}} & \hat D_{21}^\top - \Lambda \hat D_{12}\\
       \hat C_2 & \hat D_{21}- \hat D_{12}^\top\Lambda &-\sigma I+\hat D_{22}+\hat D_{22}^\top.
   \end{bmatrix} \\
   &\; -
   \begin{bmatrix}
       \hat A^\top\\
       \hat B_1^\top\\
       \hat B_2^\top 
   \end{bmatrix}P\begin{bmatrix}
       \hat A^\top\\
       \hat B_1^\top\\
       \hat B_2^\top 
   \end{bmatrix}^\top-\eta\begin{bmatrix}
       \hat C_2^\top\\
       \hat D_{21}^\top\\
       \hat D_{22}^\top
   \end{bmatrix}\begin{bmatrix}
       \hat C_2^\top\\
       \hat D_{21}^\top\\
       \hat D_{22}^\top
       \end{bmatrix}^\top\succ 0.
\end{align}
Then, Claim 2 follows by a similar procedure as Claim 1.
\end{proof}

\subsection{Parameterization of Orthogonal Layers and BiLipREN}
  
The model parameterization of static/dynamic orthogonal layers relies on a differentiable parameterization of orthogonal weights $Q\in \R^{m\times m}$. One approach is based on the Cayley transformation \cite{trockman2021orthogonalizing}:
\begin{equation}\label{eq:cayley}
    Q =\mathrm{Cayley}(Z):= (I+J)(I-J)^{-1}
\end{equation}
where $J = Z^\top-Z$ and $Z\in \R^{m\times m}$ is a free parameter. 
Note that $\det(Q)=1$ for all free $Z$. To obtain orthogonal matrices with determinant of -1, one can augment the above parameterization with a Householder transformation: $   \hat{Q}=Q\bigl(I-2z z^\top\bigr)$
where $z=\hat z/|\hat z|$ with $\hat z\in\mathbb{R}^m$ as a learnable parameter. 

We now construct the a contracting and bi-Lipschitz state-space system \eqref{eq:system} via deep composition of simpler layers:
\begin{equation}\label{eq:bilip-dyn}
    \G = \Oc_{K}\circ\M_K\circ\cdots\circ\Oc_{1}\circ\M_{1}\circ\Oc_0
\end{equation}
where each $\M_k$ is a strongly I/O monotone REN and each $\Oc_k$ is a static orthogonal layer. We refer this construction as a $K$-layer BiLipREN since the composition of REN and static orthogonal layers is also a REN. By Theorems~\ref{thm:monotone-system} and \ref{thm:main}, we have that $\G$ is robustly invertible and bi-Lipschitz. Specifically, if $\G_k$ is $(\mu_k, \nu_k)$ bi-Lipschitz, then $\G$ is $(\mu,\nu )$ bi-Lipschitz with $\mu = \Pi^K_{k=1}\mu_k$ and  $\nu = \Pi^K_{k=1}\nu_k$. The inverse of a BiLipREN $\G$ is also a BiLipREN, which is given by
\begin{equation}\label{eq:bilip-dyn-inv}
    \G^{-1} = \Oc_0^{-1}\circ\M^{-1}_1\circ\Oc^{-1}_{1}\circ\cdots\circ\M^{-1}_{K}\circ\Oc^{-1}_{K}.
\end{equation}

\begin{remark}
    An analogous of construction was proposed for static (feedforward) bi-Lipschitz neural networks in \cite{wangmonotone}.
\end{remark}
\begin{remark} 
A basic example arises from the singular value decomposition (SVD) of an invertible matrix $W$, namely $W=USV$, where the matrices $U,V$ are orthogonal and the positive diagonal matrix $S$ is strongly I/O monotone.
\end{remark}
\begin{remark}
As another classical example, consider a single-input single-output linear system $\bm G(z)$ which is minimum-phase and bi-proper. Since all poles and zeros are strictly inside the unit circle, the phase $\angle \bm G(e^{j\omega})$ is bounded on $\omega\in[0,2\pi]$ and by the argument principle does not wrap, i.e. has net change of zero over $\omega\in[0,2\pi]$. Hence for a sufficiently large integer $N$, $\bm G(z)^{1/N}$ is strictly positive real: $\angle \bm G(e^{j\omega})\in(-\pi/2,\pi/2)$, and hence is strongly input-output monotone. Therefore $G$ can be represented as a product of $N$ strongly input-output monotone layers: $\bm G = \bm G^{1/N}\circ \cdots \circ \bm G^{1/N}$, a special case of \eqref{eq:bilip-dyn}. Note that $\bm G^{1/N}$ will not generally be rational, but it can be approximated arbitrarily closely as such.
\end{remark}

\section{Illustrative Examples}

In the following sections, we will present experiments which explore the robust utility of robust invertibility and the proposed BiLipREN model. All experiments were performed with an NVIDIA GeForce RTX 4090. The Python/JAX code is available at \url{https://github.com/acfr/BiLipREN}.

\subsection{Robust Inversion}

\begin{figure}[!bt]
\centering
    \includegraphics[width=0.95\linewidth]{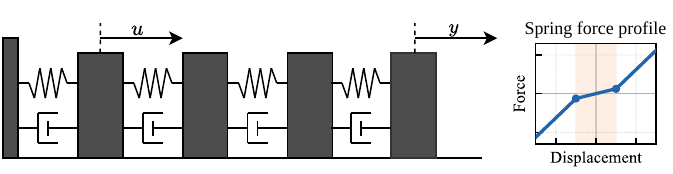}
    \caption{Schematic plot of coupled mass--spring--damper system with nonlinear spring force profile.}\label{fig:msd} 
\end{figure}

\begin{figure}[!bt]
\centering
    \includegraphics[width=\linewidth]{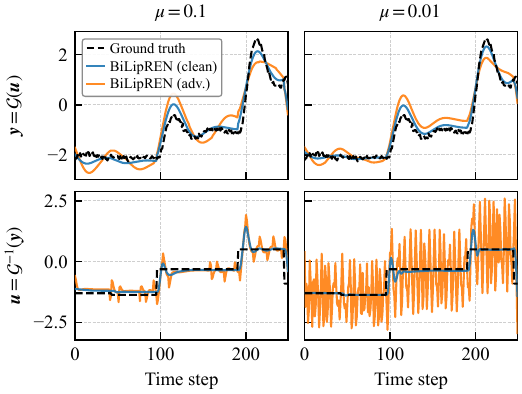}
    \caption{Forward (Top) and inverse (Bottom) response of BiLipREN without disturbance and with adversarial attacks.}\label{fig:msd-adv} 
\end{figure}

We first consider the problem of learning a robustly invertible model for the nonlinear mechanical system illustrated in Figure~\ref{fig:msd}, consisting of four coupled mass--spring--damper subsystems with nonlinear spring force profile. We excite the first cart with a random piecewise-constant input, where the amplitude and period are uniformly sampled from $[-1.5, 1.5]$ and $[0,50]$, respectively. The output is the position of the last cart, corrupted by Gaussian noise $\mathcal{N}(0,0.05)$. The training dataset consists of $200$ input-output trajectories, each of length $T=500$. 

We fit a $4$-layer BiLipREN $\G$ \eqref{eq:bilip-dyn}, with each strongly I/O monotone REN layer containing $16$ states and $64$ neurons. For comparison, we train a contracting REN (C-REN) \cite{revay2024recurrent} with similar size as BiLipREN. We then evaluate the forward model $\G$ with normalized simulation error (NSE):
\begin{equation}\label{eq:nse}
    \mathrm{NSE}(\hat{\yb},\yb)=
\frac{\lVert \hat{\yb}-\yb\rVert_T}
{\lVert \yb\rVert_T}
\end{equation}
with $\hat{\yb}=\G(\ub)$, where $(\ub,\yb)$ denotes the test input-output trajectory. The model robustness is measured by the NSE under  adversarial attacks: 
\begin{equation}
   \begin{split}
     \max_{\boldsymbol{\Delta}\ub}\; &\ \operatorname{NSE}\!\left(\mathcal{G}(\ub+\boldsymbol{\Delta}\ub),\,\yb\right) \\
     \text{s.t.}\quad& |\Delta u_t|\ \le\ \delta,\quad 0\leq t\leq T
\end{split} 
\end{equation}
where $\delta>0$ is the attack budget. For the inverse model $\G^{-1}$, the clean and adversarial NSE can be defined in a similar way.

Table~\ref{tab:msd-nse} reports the forward and inverse fitting performance under clean and adversarial inputs. Although BiLipREN imposes stronger structural constraints than C-REN, it achieves comparable performance on the forward modeling task. In contrast to C-REN, BiLipREN can yield an explicit inverse model without any additional training.

As the Lipschitz bounds are relaxed (i.e., smaller $\mu$ and larger $\nu$), the forward BiLipREN achieves improved NSE under both clean and adversarial inputs. For the inverse model, the performance remains largely unchanged except when $\mu=0.01$, where the NSE degrades significantly in the presence of adversarial attacks. As illustrated in the bottom-right panel of Figure~\ref{fig:msd-adv}, small attacks in $\yb$ can induce large oscillations in the reconstructed input  $\ub$. This highlights a trade-off in selecting $\mu$: a smaller $\mu$ improves the forward fitting accuracy but increases the inverse Lipschitz bound $1/\mu$, making the inverse model more sensitive to perturbations. 

\begin{table}[!tb] 
\centering \scriptsize \setlength{\tabcolsep}{3.0pt} \renewcommand{\arraystretch}{1.12} \begin{tabular}{c|c|c|cc|cc} \toprule \multicolumn{3}{c|}{Model} & \multicolumn{2}{c|}{Forward NSE $\ub \mapsto\yb$} & \multicolumn{2}{c}{Inverse NSE $\yb \mapsto\ub$} \\ \cmidrule(lr){1-3} \cmidrule(lr){4-5} \cmidrule(lr){6-7} Arch. & $\mu$ & $\nu$ & Clean $\downarrow$ & Adv. $\downarrow$ & Clean $\downarrow$ & Adv. $\downarrow$ \\ \midrule C-REN & -- & -- & 0.152 & 0.373 & -- & -- \\ \midrule \multirow{8}{*}{BiLipREN} & \multirow{4}{*}{$0.1$} & $4$ & 0.281 & 0.362 & 0.418 & 0.445 \\ & & $8$ & 0.216 & 0.321 & 0.415 & 0.452 \\ & & $16$ & 0.193 & 0.304 & 0.426 & 0.468 \\ & & $32$ & 0.181 & 0.298 & 0.422 & 0.466 \\ \cmidrule(lr){2-7} & $0.01$ & \multirow{4}{*}{$8$} & 0.179 & 0.287 & 0.381 & 0.610 \\ & $0.05$ & & 0.198 & 0.308 & 0.404 & 0.463 \\ & $0.1$ & & 0.216 & 0.321 & 0.415 & 0.452 \\ & $0.5$ & & 0.298 & 0.390& 0.399 & 0.411 \\ \bottomrule \end{tabular} \caption{Forward and inverse model fitting performance under clean and adversarial inputs, where the attack budget is $\delta=0.05$. } \label{tab:msd-nse} \end{table}

\subsection{Nonlinear Inner-outer Factorization} 
We consider the problem of learning a nonlinear inner–outer factorization \eqref{eq:inner-outer} for a stable nonlinear delay system 
\begin{equation}\label{eq:deley-siso}
    y_{t+1}=0.9 \tanh (y_t)+u_{t-3}.
\end{equation}
The training dataset consists of $2000$ input–output trajectories, where each of length $50$ is generated using inputs $u_t\sim \mathcal{N}(0,1)$ for $t\geq 0$ and $u_t=0$ for $t<0$. We fit a BiLipREN which is deliberately over-parameterized to illustrate robustness, consisting of a $10$-state dynamic orthogonal layer $\bm{O}(z)$ without bias term, and a strongly I/O monotone REN $\M$ with $4$ states and $16$ neurons as the robustly invertible model. The Lipschitz parameters of $\M$ are chosen as $\mu = 0.1$ and $\nu=8$. 

The impulse responses in Figure~\ref{fig:io-factorization} (top) demonstrate that the learned dynamic orthogonal layer closely approximates the 3-step time delay in \eqref{eq:deley-siso}. Meanwhile, the outer component $\mathcal{G}$ captures the associated minimum-phase dynamics of \eqref{eq:deley-siso}, as illustrated in Figure~\ref{fig:io-factorization} (bottom).

\begin{figure}[!tb]
\centering
    \includegraphics[width=0.85\linewidth]{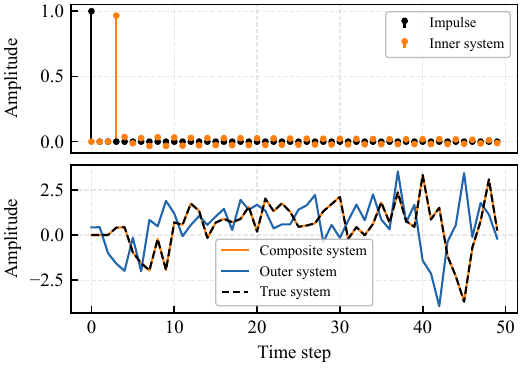}
    \caption{The impulse response of the inner system (Top). The responses of the outer system, the original system and the composed system under random inputs (Bottom).}\label{fig:io-factorization} 
\end{figure}

\section{Internal Model Control}
Tracking and disturbance rejection are central objectives in control design, and Internal model control (IMC) is a classical design method for these objectives which can explicitly utilize a stable model inverse inside a feedback loop \cite{garcia1982internal}. Here we illustrate the use of robustly invertible nonlinear dynamics for a nonlinear version of this classical scheme, in which the aim is that the response to reference commands $\bm{r}\mapsto\bm{y}$ closely follows a contracting and Lipschitz desired response $\mathcal R$. In the classical setting $\mathcal R$ is often a simple low pass filter, e.g. $\bm{R}(z) = \tfrac{(1-p)z}{z-p}\circ I$.

\begin{figure}[!tb]
    \centering
    \includegraphics[width=0.85\linewidth]{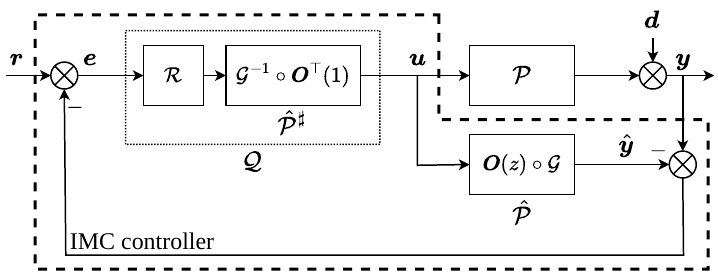}
    \caption{Internal model control structure based on learned nonlinear inner-outer factorization.}
    \label{fig:imc-close}
\end{figure}

\begin{prop}\label{prop:1}
Consider the IMC scheme shown in Figure~\ref{fig:imc-close}. Suppose that $\mathcal{P}$ is contracting and Lipschitz, and $\hat{\mathcal{P}}$ admits a nonlinear inner--outer factorization $\hat{\mathcal{P}}=\bm{O}(z)\circ\G$, where $\G$ is $(\mu,\nu)$ bi-Lipschitz and $\bm{O}(z)$ is a dynamic orthogonal layer. Suppose the desired response $\mathcal R$ is contracting and $\omega$-Lipschitz 
\begin{enumerate}
\item Without model mismatch (i.e. $\Delta\mathcal{P}:=\mathcal{P}-\hat{\mathcal{P}}=0$), the closed loop is contracting with response:
\begin{equation}
    \begin{split}
        \yb &=\bm{O}(z)\circ\bm{O}(1)^\top \circ \mathcal{R} (\bm{r}-\bm{d})+ \bm{d},\\
        \ub &=\hat{\mathcal{P}}^{\sharp}\circ \mathcal{R} (\bm{r}-\bm{d}) =\mathcal{G}^{-1} \circ\bm{O}(1)^\top \circ \mathcal{R} (\bm{r}-\bm{d}).
    \end{split}\label{eq:IMC_response}
\end{equation}
\item With model mismatch, suppose $\Delta\mathcal{P}$ is $\gamma$-Lipschitz  with $\gamma \omega<\mu$, then the closed-loop system is contracting.
\end{enumerate}
In both cases, if the steady-state map of $\mathcal{R}$ is the identity, then the closed loop achieves zero steady-state error to step references and disturbances.
\end{prop}

\begin{remark} In the absence of model error and disturbances, the response $\bm r \mapsto \bm y$  is the desired response $\mathcal R$ composed with the non-invertible dynamic component $\bm O(z) \circ \bm O(1)^\top$ of the all-pass (inner) factor $\bm O(z)$. 
\end{remark}

\begin{remark}
    The proposed IMC controller is a special case of the Youla-REN policy class in \cite{barbara2025react}, with $\hat{P}$ serving as the system observer and $\mathcal{Q}=\hat{\mathcal{P}}^{\sharp}\circ \mathcal{R}$ as the Youla parameter. 
\end{remark}

\begin{proof}
Claim 1): Without model mismatch, the error $\bm{e}=\bm{r}-\bm{d}$ becomes an external signal, so the closed loop is contracting by construction. The closed-loop input response is $\ub=\hat{\mathcal{P}}^{\sharp}\circ \mathcal{R}(\bm{r}-\bm{d})$ and the output response is 
\[
\begin{split}
    \yb&=\bm{d}+\mathcal{P}(\ub)= \bm{d}+\hat{\mathcal{P}}\circ\hat{\mathcal{P}}^{\sharp}\circ \mathcal{R}(\bm{r}-\bm{d}) \\
    &=\bm{d}+\bm{O}(z)\circ\bm{O}(1)^\top \circ \mathcal{R} (\bm{r}-\bm{d}).
\end{split}
\]
If $\mathcal{R}_{ss}=I$, then the steady-state output $y_{ss}$ satisfies
\[
y_{ss}=d_{ss}+\bm{O}(1)\bm{O}(1)^\top \mathcal{R}_{ss} (r-d_{ss})=r.
\]
Claim 2): Since $\mathcal{R}$ is $\omega$-Lipschitz, we obtain that the loop gain of $\mathcal{Q}:=\hat{\mathcal{P}}^{\sharp}\circ \mathcal{R}$ and $\Delta \mathcal{P}$ satisfies
\begin{equation}\label{eq:small-gain}
    \mathrm{Lip}(\mathcal{Q})\mathrm{Lip}(\Delta \mathcal{P})\leq \mathrm{Lip}(\hat{\mathcal{P}}^{\sharp})\mathrm{Lip}(\Delta\mathcal{P})\leq \frac{\gamma \omega}{\mu}<1.
\end{equation}
where $\mathrm{Lip}(\cdot)$ denotes the least Lipschitz upper bound of an operator. Then, the closed loop is contracting by incremental small gain theorem, implying that it converges to a unique steady state under constant inputs. Since $\mathcal{R}_{ss}=I$, we have
\begin{equation}
    \hat{y}_{ss}=\hat{\mathcal{P}}_{ss}(u_{ss})=\hat{\mathcal{P}}_{ss}\circ \hat{\mathcal{P}}_{ss}^{\sharp} \circ \mathcal{R}_{ss}(e_{ss})=e_{ss},
\end{equation}
where $\hat{\mathcal{P}}_{ss}$ and $\hat{\mathcal{P}}_{ss}^{\sharp}$ denote the steady state maps of $\hat{\mathcal{P}}$ and $\hat{\mathcal{P}}^{\sharp}$, respectively. This further implies  $y_{ss}=\hat{y}_{ss}+r-e_{ss}=r$.
\end{proof}

\subsection{Numerical Example}
\begin{figure}[!tb]
    \centering
    \includegraphics[width=\linewidth]{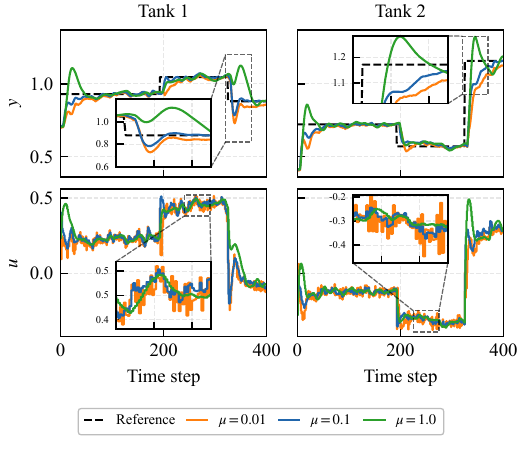} 
    \caption{Closed-loop responses of the learned IMC controllers under different lower Lipschitz bounds, showing the output tracking performance (Top), and the corresponding normalized control inputs (Bottom).}
    \label{fig:imc-response}
\end{figure}

We implement the proposed learning based IMC approach on a time-delayed quadruple tank system from \cite{bonassi2022recurrent}:
\[
\begin{aligned}
\dot{h}_1(t) &= -\frac{a_1}{S}\sqrt{2gh_1(t)} + \frac{a_3}{S}\sqrt{2gh_3(t)} + \frac{\gamma_a}{S}q_a(t-\tau) \\
\dot{h}_2(t) &= -\frac{a_2}{S}\sqrt{2gh_2(t)} + \frac{a_4}{S}\sqrt{2gh_4(t)} + \frac{\gamma_b}{S}q_b(t-\tau) \\
\dot{h}_3(t) &= -\frac{a_3}{S}\sqrt{2gh_3(t)} + \frac{1-\gamma_b}{S}q_b(t-\tau) \\
\dot{h}_4(t) &= -\frac{a_4}{S}\sqrt{2gh_4(t)} + \frac{1-\gamma_a}{S}q_a(t-\tau)
\end{aligned}
\]
where $h_i$ is $i$th tank water level, $q_a,q_b$ are the input flow rates. We take $\gamma_a=0.5$, $\gamma_b=0.6$ and $\tau=125\mathrm{s}$. The values of $S$ and $a_i$ can be found in   \cite{bonassi2022recurrent}. The system output is $y=[h_1,h_2]^\top$ and the control input is $u=[u_1,u_2]^\top$ where 
\begin{equation*}
    \begin{split}
        q_a&=\frac{q_{a,\min} + q_{a,\max}}{2}+\frac{q_{a,\max}-q_{a,\min}  }{2}u_1, \\
        q_b&=\frac{q_{b,\min} + q_{b,\max}}{2}+\frac{q_{b,\max}-q_{b,\min}}{2}u_2.
    \end{split}
\end{equation*}
The system is discretized with sampling period of $25\mathrm{s}$.  We generated $1000$ input-output trajectories by simulating the system with a pseudo-random multilevel signal whose amplitude is drawn uniformly over the pump's operating range and held for a random duration $[2, 12]$. The training dataset is constructed with randomly extracted input-output sequences of length $T=50$ from these trajectories. 

We parameterize $\hat{\mathcal{P}}$ via the nonlinear inner-outer factorization \eqref{eq:inner-outer}, which consists of an 8-state dynamic orthogonal layer $\bm O(z)$ without a bias term and a $3$-layer $(\mu,\nu)$ BiLipREN $\G$ \eqref{eq:bilip-dyn} with each strongly I/O monotone REN layer consisting of 6 states and 16 neurons. We take a fixed Lipschitz bound $\nu=10$ and difference choices of inverse Lipschitz bound $\mu$. We then simulated the closed-loop system with the learned IMC controller and low-pass filter $\mathcal{R}$ under random piece-wise reference $\bm{r}$ and Gaussian noise $d_t\sim \mathcal{N}(0,0.05)$. 

Table~\ref{tab:imc} reports the performance of model fitting and IMC control performance, measured by NSE \eqref{eq:nse} and root mean square error (RMSE) $\|\yb-\bm{r}\|/\sqrt{T}$, respectively. We estimate the empirical Lipschitz constant of $\Delta\mathcal{P}$ by randomly sampling a large number of input-trajectory pairs, simulating their corresponding residual outputs, and computing the maximum observed input-output difference ratio. The results show that the (generally conservative) small gain condition \eqref{eq:small-gain} is satisfied for all cases except $\mu=0.01$, although the closed-loop system did not exhibit unstable behavior. 

As shown in Table~\ref{tab:imc}, reducing $\mu$ improves model fitting so that the closed-loop response of the control input approximates \eqref{eq:IMC_response}
which has an Lipschitz bound of $\omega/\mu$. Thus, a too-small $\mu$ can result in $\ub$ being more sensitive to noise $\bm d$, see the second row of Figure~\ref{fig:imc-response}. Larger $\mu$ results in poorer model fits performance and larger empirical Lipschitz bound of $\Delta\mathcal{P}$, but produces smoother control actions and improved tracking while still satisfying the small-gain condition~\eqref{eq:small-gain}. An excessively large $\mu$, such as $\mu=1.0$, increases the model mismatch and causes significant overshoot, thereby degrading the closed-loop performance.

\begin{table}[!tb]
\centering
\scriptsize
\setlength{\tabcolsep}{4pt}
\renewcommand{\arraystretch}{1.12}
\begin{tabular}{c|c|c|c}
\toprule
Learned I/O fact. $\hat{\mathcal P}$
& Model mismatch $\Delta\mathcal P$
& Model fit
& Closed-loop \\
Cert. inv. Lip. $\mu$
& Emp. Lip. $\gamma$
& NSE $\downarrow$
& RMSE $\downarrow$ \\
\midrule
$0.01$ & $0.029$ & $0.045$ & $0.070$ \\
$0.05$ & $0.038$ & $0.057$ & $0.051$ \\
$0.1$  & $0.049$ & $0.074$ & $0.049$ \\
$0.5$  & $0.179$ & $0.274$ & $0.045$ \\
$1.0$  & $0.577$ & $0.867$ & $0.060$ \\
\bottomrule
\end{tabular}
\caption{Fitting NSE and closed-loop RMSE under different certified
inverse Lipschitz bounds $\mu$ and the corresponding empirical
model-mismatch Lipschitz constants $\gamma$.}
\label{tab:imc}
\end{table}

\section{Surrogate Dynamic Loss Learning}
Consider a standard-form trajectory optimization problem:
\begin{equation}\label{eq:ocp}
\begin{aligned}
\min_{\ub_{[T]}\in\ell^m}\quad
&J\left(\ub_{[T]}\right)
:=
c_f(x_{T+1})+\sum_{t=0}^{T}c_t(x_t,u_t)\\
\mathrm{s.t.}\quad
&x_{t+1}=f(x_t,u_t), \quad x_0=a
\end{aligned}
\end{equation}
where $T$ is the horizon, $c_t,c_f$ are the stage and terminal cost functions, respectively. 

In this work, we consider a ``black-box'' setting in which the model $f$, initial state $a$, and cost functions $c_t,c_f$ are unknown. Instead, we only have access to a finite dataset of sampled input-loss pairs $\{(\ub^i, J^i)\}_{1\leq i\leq N}$. Our approach consists of two steps:
\begin{enumerate}
    \item Learn a differentiable surrogate loss $\hat{J}:\ell^m\to \R$ to fit the  dataset, i.e. $\hat J(u_i)\approx J_i$ for all $i$. 
    \item Compute an control input sequence by solving the surrogate optimization problem:
    \begin{equation}\label{eq:surrogate-opt}
        \hat{\ub}_{[T]}^\star:=\argmin_{\ub_{[T]}\in \ell^m}\; \hat{J} \left(\ub_{[T]}\right).
    \end{equation}
\end{enumerate}
Although we don't pursue this here, these two steps may be performed iteratively, with the surrogate model being refined as new candidate solutions $\hat{\ub}_{[T]}^\star$ are generated and evaluated. 

Fitting a black-box neural surrogate (e.g. a neural network $\hat{J}:\R^{m(T+1)}\rightarrow \R$) has two issues: firstly, neural network structures require the time horizon to be fixed, and model complexity grows with the horizon. Secondly, while the surrogate may achieve a good fit to the training data, it can introduce undesirable optimization artifacts such as spurious local minima, flat regions, or poorly conditioned gradients. 

To address the first issue, we introduce surrogate loss functions of the form:
\begin{equation}\label{eq:dyn-pl}
\hat{J}(\ub_{[T]})=\frac{1}{2}\|\G(\ub_{[T]})\|_{T}^{2}+c 
\end{equation}
where $\G$ is a dynamical system and $c\in \R$ is a learnable bias term, capturing the temporal structure in the problem. To address the second problem, parameterize $\G$ as a $(\mu,\nu)$ bi-Lipschitz recurrent model \eqref{eq:bilip-dyn}.

Analogously to the finite-dimensional case in \cite{wangmonotone}, we can show that the proposed surrogate loss satisfies the following PL condition (\cite{polyak1963gradient,lojasiewicz1963topological}):
\begin{equation}
    \frac{1}{2}\|\nabla \hat{J}(\ub)\|_T^2\geq \mu \left[\hat{J}(\ub_{[T]})-\hat{J}(\hat{\ub}_{[T]}^\star)\right],\;\forall \ub_{[T]}\in \ell^m.
\end{equation}
The PL condition is significant in optimization as it is weaker than convexity, but still implies gradient methods converge to a global minimum $\hat{\ub}_{[T]}^\star$ with a linear rate. Moreover, due to the quadratic structure  \eqref{eq:dyn-pl} and robust invertibility of $\G$, the surrogate loss $\hat{J}$ has a unique  minimizer $\hat{\ub}_{[T]}^\star$ given by
\begin{equation}\label{eq:plnet-inv}
  \hat{\ub}_{[T]}^\star=\G^{-1}(\bm 0)_{[T]},  
\end{equation}
which provides a direct, non-gradient-based method for solving \eqref{eq:surrogate-opt}. 

\subsection{Numerical Example}\label{sec:car-traj}
\begin{figure*}[!tb]
    \centering
    \includegraphics[width=\textwidth]{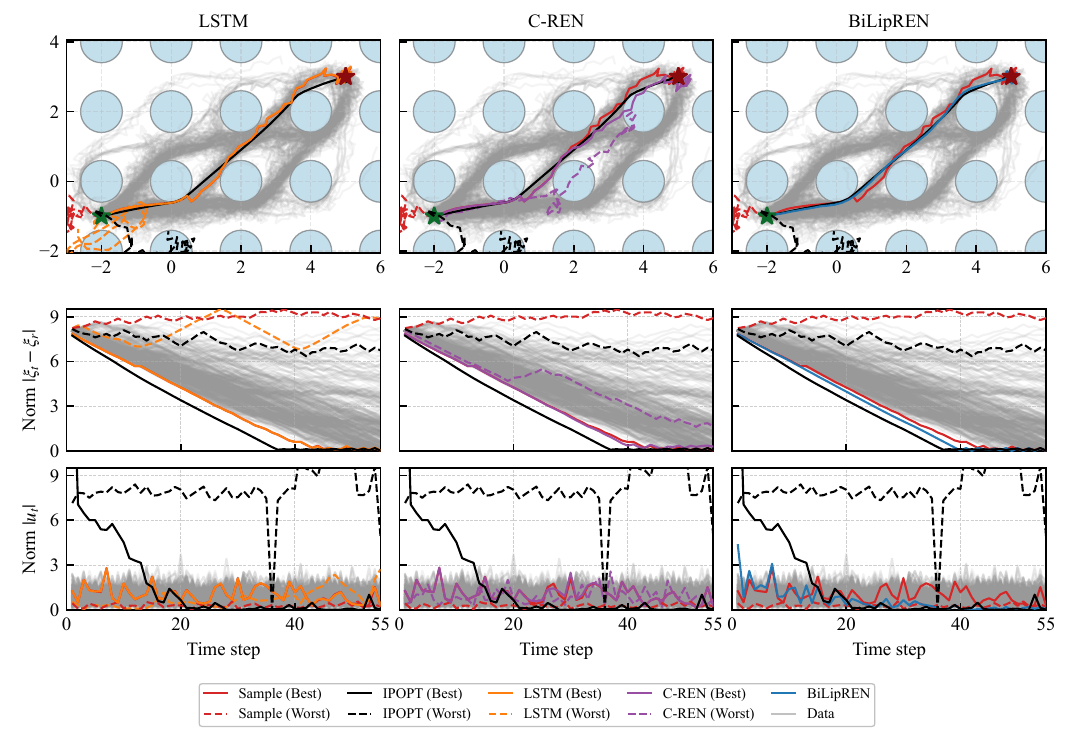}
    \caption{Trajectory samples and the corresponding input sequences obtained by minimizing the surrogate loss model \eqref{eq:dyn-pl}. Results are shown for three parameterizations of $\G$: LSTM, C-REN, and BiLipREN, compared to  IPOPT, which has access to the true model and cost.}
    \label{fig:mbd-pl}
\end{figure*}

We consider a simple 2D obstacle-avoidance problem from \cite{mishra2025eb}. The objective is to generate an input sequence that steers the vehicle from an initial position to a target position while avoiding obstacles, as illustrated in Figure~\ref{fig:mbd-pl}. The vehicle dynamics and cost function are given by
\begin{equation}
\begin{aligned}
\xi_{t+1} &= \xi_t + 0.4\,\mathrm{sigmoid}(|u_t|)\,\frac{u_t}{|u_t|},\quad \xi_0=\xi_s, \\
J(\ub_{[T]}) &= 20\left|\xi_{T}-\xi_r\right|
+ 10\sum_{t=0}^{T-1}\bigl[\left|\xi_t-\xi_r\right|
+ 0.1\left|u_t\right| \\
&+\mathbf{1}_{\mathbb{X}}(\xi_t)\bigr],
\end{aligned}
\end{equation}
where $\xi_t, u_t\in \R^2$ denote the state and input, respectively, $\xi_s$ and $\xi_r$ are the initial and target positions. The sigmoid function is defined as $\mathrm{sigmoid}(x) = 1/(1 + e^{-x})$ and $\mathbf{1}_{\mathbb{X}}$ is the indicator function of the obstacle set $\mathbb{X}$ (i.e., the blue disks in Figure~\ref{fig:mbd-pl}). 

We generate $25,000$ trajectories, each of length $T=55$, using the emerging barrier model based diffusion (EB-MBD) method \cite{mishra2025eb}. This provides only very limited coverage of the input space, which has dimension of $2T=110$. We collect the corresponding input sequences and the cost values, which span a wide range, see Table \ref{tab:surrogate-loss}, making accurate surrogate loss learning challenging. To concentrate the fitting accuracy on high-quality samples, we introduce the following cost-weighted regression loss:
\begin{equation}\label{eq:pl-loss}
L=
\frac{\sum_{i} w^i\bigl(\hat J(\ub^i)-J^i\bigr)^2}{\sum_{i} w^i}
\end{equation}
where the sample weights are $w^i=e^{-(J^i-J_{m})/\tau}$ with $\tau>0$ denoting the temperature parameter and $J_{m}$ the minimum objective value among all training samples. The training loss is further augmented with the regularization term $\lambda(c-J_{m})^2$, which encourages the minimum surrogate loss value $c$ to remain close to $J_{m}$.  We use the dynamic surrogate dynamic model \eqref{eq:dyn-pl} where $\G$ is parameterized as a 5-layer BiLipREN \eqref{eq:bilip-dyn} with $\mu=0.1$ and $\nu=48$, each strongly I/O monotone REN layer containing $8$ states and $64$ neurons. For comparison, we also train surrogate models in which $\G$ is parameterized by a Long Short-Term Memory (LSTM) \cite{hochreiter1997long} and a C-REN of comparable model size. 

After training, we solve the surrogate optimization problem \eqref{eq:surrogate-opt} using gradient descent for LSTM and C-REN, and using the exact inverse for BiLipREN. Since the surrogate loss function $\hat{J}$ is highly non-convex, we initialize the gradient methods with $60$ trajectories randomly selected from the training dataset. To provide a baseline with perfect model knowledge, we also solve the original trajectory optimization problem \eqref{eq:ocp} using IPOPT method \cite{wachter2006implementation} with the same initial guesses.

\begin{table}[!tb]
\centering
\scriptsize
\setlength{\tabcolsep}{4pt}
\renewcommand{\arraystretch}{1.12}
\begin{tabular}{c c c c}
\toprule
Model & Fitting loss $L$ & Best cost $J$ & Worst cost $J$ \\
\midrule
Dataset  & --     & 1863& 5055 \\
\midrule
LSTM     & 1718   & 1868 & 4758 \\
C-REN    & 6014  & 1918 &2996\\
BiLipREN & 22805 & 1672 & --\\
\midrule
IPOPT    & --     & 1618 &5837\\
\bottomrule
\end{tabular}
\caption{The fitting error of surrogate loss model regression and the best/worst costs of input trajectories generated by surrogate optimization.}
\label{tab:surrogate-loss}
\end{table}

Table~\ref{tab:surrogate-loss} and Figure~\ref{fig:mbd-pl} show although the surrogates based on LSTM and C-REN fit the mapping $u\mapsto J$ more accurately, optimizing them via gradient descent was highly sensitive to initialization, and even their best results are quite ragged and did not improve upon the best training sample. In contrast, BiLipREN analytically produces a single solution from \eqref{eq:plnet-inv} which is smooth, avoids obstacles, and out-performs the best training sample. We note that even the baseline IPOPT, which has perfect access to the model and cost, is also highly sensitive to initialization due to the non-convex nature of the problem. 

To summarize, the proposed surrogate loss model comes equipped with an analytically-computable minimum, which in this example improved upon the best sample in the training data, indicating that it encodes some useful structure of the problem.

\section{Dynamic Generative Models}

We consider the problem of learning a generative model for a complex trajectory distribution $\yb_{[T]}\sim p_{\mathrm{data}}$. Specifically, we seek to learn a robustly invertible dynamical model $\G$ such that matching the data distribution can be generated by
\begin{equation}\label{eq:gen-model}
    \yb_{[T]}=\G\bigl(\ub_{[T]}\bigr),\quad u_t\sim \mathcal{N}(0,I).
\end{equation}
That is, the generative model $\G$ transforms a Gaussian white-noise sequence into samples with a similar distribution to the training data see the bottom-right panel of Figure~\ref{fig:comparison}. 

Analogous to the change-of-variables formula in \eqref{eq:change-of-var}, the probability density of $\yb_{[T]}$ is 
\begin{equation}
    \begin{split}
        p_{\yb}(\yb_{[T]})&=p_{\ub}(\G^{-1}(\yb_{[T]})\bigl|\det J_{\G^{-1}}(\yb_{[T]})\bigr| \\
        &= \prod_{t=0}^{T}p_u(u_t)\left|\det \frac{\partial u_t}{\partial y_t}\right|,
    \end{split}
\end{equation}
where the second equality follows from the independence of the latent variables $u_t$ and the causality of $\G^{-1}$. Following the framework of normalizing flows \cite{papamakarios2021normalizing}, we train the model by minimizing the negative log-likelihood (NLL):
\begin{equation*}
    L_{\mathrm{NLL}}=-\sum_{t=0}^{T}\left[\log p_u(u_t)+\log \left|\det \frac{\partial u_t}{\partial y_t}\right|\right].
\end{equation*}
The first term encourages the latent sequence to follow the Gaussian prior, while the second term accounts for the local volume change induced by the invertible transformation $\G^{-1}$. 

\subsection{Generative Trajectory Modeling}
To illustrate the application to generative modelling, we revisit the 2D obstacle-avoidance problem introduced in Section~\ref{sec:car-traj}. The goal is to learn a generative model that accurately reproduces the multi-modal distribution of the training set, consisting of  $10{,}000$ state trajectories.

We parameterize $\G^{-1}$ by an acyclic BiLipREN \eqref{eq:bilip-dyn} with 12 strongly I/O monotone REN layers, each comprising $10$ states and $32$ neurons. The Lipschitz parameters are $\mu=0.01$ and $\nu=20$. 

Mapping the training data back through $u= \mathcal G^{-1}(y)$ should result Gaussian white noise. This is supported by Figure~\ref{fig:qq}: the autocorrelation coefficients lie within the 95\% confidence interval, indicating negligible temporal correlation and the Q--Q plot further shows that the empirical quantiles closely align with those of a standard Gaussian distribution. For generation, we sample $200$ samples of $u$ as Gaussian white noise, and mapped them through
\eqref{eq:gen-model}. As illustrated in Figure~\ref{fig:dyn-nf-inv}, the generated trajectories capture the statistical variability and the multi-modal distribution of the training data.
\begin{figure}[!tb]
    \centering
    \includegraphics[width=\linewidth]{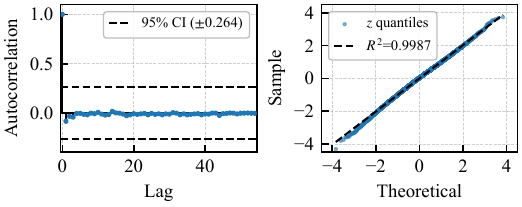}
    \caption{Autocorrelation (Left) and Q-Q plots (Right) of the latent variables produced by the BiLipREN.}
    \label{fig:qq}
\end{figure}
\begin{figure}[!tb]
    \centering
    \includegraphics[width=\linewidth]{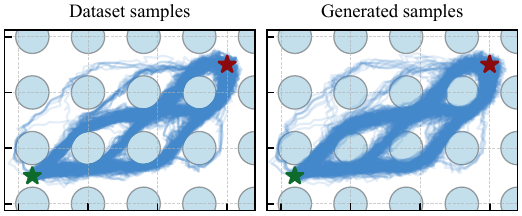}
    \caption{200 trajectories sampled from the dataset (Left) and generated (Right) via \eqref{eq:gen-model}.}
    \label{fig:dyn-nf-inv}
\end{figure}

\section{Conclusion}
In this paper, we introduced the notion of robust invertibility for nonlinear dynamical systems and provided tractable sufficient conditions through contraction and bi-Lipschitz properties. We then proposed the BiLipREN, which is a robustly invertible recurrent neural network model constructed via series composition of static orthogonal layers and strongly I/O monotone dynamic layers. The proposed model set admits a flexible and differentiable parameterization, allowing model training via standard gradient descent methods. We have illustrated the utility of the new model class on problems in internal model control, trajectory optimization, and generative modelling of signals and trajectories. 

\section*{References}
\bibliographystyle{IEEEtran}
\bibliography{ref.bib}

\end{document}